%% file: all_in.tex
\documentclass[journal]{IEEEtran}

\usepackage{indentfirst}
\usepackage{amsmath,amssymb,amsthm}
\usepackage{bm}
\usepackage{graphicx}
\usepackage{cite}
\usepackage{algpseudocode}
\usepackage{algorithm}
\usepackage{booktabs}
\usepackage[T1]{fontenc}
\usepackage{xcolor}
\usepackage{multirow}
\usepackage{subcaption}
\usepackage{hyperref}
\usepackage{fancyhdr}
\fancypagestyle{arxivnotice}{%
  \fancyhf{}%
  \fancyhead[C]{\footnotesize This work has been submitted to the IEEE JSAC for possible publication. Copyright may be transferred without notice, after which this version may no longer be accessible.}%
}
\newcommand{\tr}{\operatorname{tr}}

\newcommand{\Real}{\operatorname{Re}}
\newcommand{\rank}{\operatorname{rank}}
\newcommand{\vect}[1]{\mathbf{#1}}
\newcommand{\mat}[1]{\mathbf{#1}}
\newcommand{\FIM}{\mathbf{I}}
\newcommand{\JFIM}{\mathbf{J}}   % Joint 4x4 FIM
\newcommand{\CRB}{\mathrm{CRB}}
\usepackage{titlesec}
\titlespacing*{\section}{0pt}{8pt plus 2pt minus 2pt}{3pt plus 1pt minus 1pt}
\titlespacing*{\subsection}{0pt}{6pt plus 2pt minus 2pt}{2pt plus 1pt minus 1pt}
\titlespacing*{\subsubsection}{0pt}{4pt plus 1pt minus 1pt}{1pt plus 1pt minus 1pt}
\newtheorem{theorem}{Theorem}
\newtheorem{proposition}[theorem]{Proposition}

\newtheorem{lemma}[theorem]{Lemma}
\newtheorem{remark}{Remark}
\newtheorem{definition}{Definition}

\begin{document}
\bstctlcite{IEEEexample:BSTcontrol} %this one using after begin document
%\title{Joint Position--Velocity CRB Analysis for FIM-Augmented Cell-Free ISAC Systems}

%\title{Twin-in-the-Loop Closed-Loop Optimization of Position–Velocity Estimation in FIM-Augmented Cell-Free ISAC Systems}

\title{Twin-in-the-Loop Optimization and Fundamental Limits of Position--Velocity Estimation in Cell-Free ISAC Systems}

%optional one: Digital Twin–Driven Closed-Loop Optimization for Joint Position–Velocity Estimation in Cell-Free ISAC Systems
%Predictive Digital Twin–Driven Sensing: Closed-Loop Position–Velocity CRB Optimization in Cell-Free ISAC Networks
%\author{
%\IEEEauthorblockN{Changhao He\textsuperscript{1}， Xiaojuan %Zhang\textsuperscript{2}}

%\IEEEauthorblockA{\textsuperscript{1}King Abdullah University %of Science and Technology (KAUST), Saudi Arabia}

%\IEEEauthorblockA{\textsuperscript{2}Institute for Infocomm Research, Agency for Science, Technology and Research (A*STAR), Singapore}

% \IEEEauthorblockA{\textsuperscript{3}[Third Institution]}
%}

% \author{
% \IEEEauthorblockN{Changhao He, Xiaojuan Zhang,~\IEEEmembership{Senior Member,~IEEE}}
% \thanks{C. H. He is with King Abdullah University of Science and Technology (KAUST), Saudi Arabia (e-mail: changhao.he@kaust.edu.sa
% ). X. J. Zhang is with the Institute for Infocomm Research, Agency for Science, Technology and Research (A*STAR), Singapore (e-mail: xiaojuanzhang@ieee.org
% ). Corresponding author: X. J. Zhang.}}
\author{
\IEEEauthorblockN{Changhao He, Xiaojuan Zhang,~\IEEEmembership{Senior Member,~IEEE}, Geoffrey Ye Li,~\IEEEmembership{Fellow,~IEEE}}
\thanks{C. H. He is with King Abdullah University of Science and Technology, Saudi Arabia (e-mail: changhao.he@kaust.edu.sa).
X. J. Zhang is with the Institute for Infocomm Research, Agency for Science, Technology and Research, Singapore (e-mail: xzhang@a-star.edu.sg).
G. Y. Li is with the Department of Electrical and Electronic Engineering, Imperial College London, London, U.K. (e-mail: geoffrey.li@imperial.ac.uk).
Corresponding author: X. J. Zhang.}}

\maketitle
\thispagestyle{arxivnotice}

\input{00_abstract}
\input{01_introduction}
\input{02_system_model}
\input{03_velocity_crb}
\input{04_joint_crb}
\input{05_mismatch_robust}
\input{06_dt_loop}
\input{07_results}

\input{08_conclusion}

% ─── References ───
\bibliographystyle{IEEEtran}
\bibliography{IEEEabrv,references}
% ─── Appendices ───
% \appendices
% \input{sections/09_appendices}
\input{Appendix}
\end{document}

%% file: 00_abstract.tex
\begin{abstract}
    Digital twin (DT) networks require tight integration with wireless sensing, yet the fundamental limits of such coupling in cell-free integrated sensing and communication (ISAC) systems remain largely unexplored, particularly in the presence of fluid intelligent metasurfaces (FIM). This paper establishes a joint position--velocity Cramér–Rao bound (CRB) framework, operationalized through a twin-in-the-loop architecture. By leveraging a scatter-matrix decomposition of the velocity Fisher information, we show that single-base-station systems are inherently rank-deficient for two-dimensional velocity estimation, whereas cell-free deployments with multiple access-point pairs achieve full observability. The resulting CRB reveals a spatio-temporal decoupling: FIM shape optimization significantly improves position accuracy but does not affect the velocity CRB under isotropic waveforms while Doppler coupling asymmetrically enhances position estimation accuracy. Building on this analysis, we develop a closed-loop DT framework, deriving the critical mismatch angle in closed form and showing that angular diversity in cell-free systems mitigates DT prediction errors. We further characterize the optimal synchronization period and propose a confidence-aware scheduling strategy that reduces the DT update rate. Numerical results demonstrate substantial performance gains over single-base-station systems, with improvements attributed to angular diversity, Doppler–-position coupling, and FIM adaptation.
\end{abstract}


%% file: 01_introduction.tex
% sections/01_introduction.tex
\section{Introduction}
\label{sec:intro}
\IEEEPARstart{D}{igital} twin (DT) technology, originally developed for cyber-physical systems, is emerging as a key enabler of \emph{closed-loop, application-aware wireless networks}~\cite{Hakiri2024DTSurvey,Zhang2024DNTSurvey}. Rather than passively mirroring the physical environment, next-generation DTs act as \emph{active agents in the control loop}: a predictive twin continuously estimates future network and environment states, the wireless system adapts its configuration accordingly, and real-time sensing feedback updates the twin, forming a tightly coupled predict--optimize--sense--update cycle~\cite{jiang2024dt_survey}. 
This paradigm shift positions wireless infrastructure as an adaptive substrate for scalable digital twin networks (DTNs), enabling real-time decision-making in dynamic environments.

This closed-loop capability is particularly compelling for integrated sensing and communication (ISAC)~\cite{liu2022isac,liu2020crb_dfrc}, where sensing measurements play a dual role: they not only extract target information but also continuously refine the twin’s representation of the propagation environment. As a result, sensing performance directly impacts the fidelity of the DT, which in turn governs future sensing and communication strategies, creating a tightly coupled feedback loop that demands joint system design.

Within this context, the cell-free architecture, where geographically distributed access points (APs) jointly serve users and sense targets under central coordination, offers a key structural advantage: \emph{angular diversity}~\cite{zhang2011csitr,bjornson2020cellfree_book}.
By observing targets from multiple spatially separated viewpoints, cell-free systems significantly enhance environmental observability and robustness compared to co-located arrays.
Recent works have demonstrated the potential of cell-free ISAC for multi-static sensing~\cite{demirhan2024cellfree_isac,behdad2024multistatic,buzzi2021cellfree}, and our prior work~\cite{he2026globecom} showed that cell-free angular diversity can amplify the morphing gain of reconfigurable arrays by more than an order of magnitude over single-BS deployments.

% Within this context, the cell-free architecture, where geographically distributed access points (APs) jointly serve users and sense targets under central coordination, offers a key structural advantage: \emph{angular diversity}~\cite{zhang2011csitr,bjornson2020cellfree_book}. 
% By observing targets from multiple spatially separated viewpoints, cell-free systems significantly enhance environmental observability and robustness compared to co-located arrays. 
% Recent works have demonstrated the potential of cell-free ISAC for multi-static sensing~\cite{demirhan2024cellfree_isac,behdad2024multistatic,buzzi2021cellfree}, and our prior work~\cite{he2026globecom} showed that such angular diversity substantially amplifies the gain of reconfigurable sensing architectures.

A second key enabler is the fluid intelligent metasurface (FIM)~\cite{wong2020fluid,zhu2024movable}, which allows continuous repositioning of antenna elements within a given aperture. 
This capability introduces a new degree of freedom for \emph{hardware-level adaptation}, enabling the wireless system to reshape its spatial sensing geometry in response to DT predictions. 
When integrated into a DT-driven ISAC framework, the predicted target state can directly guide FIM configuration, enabling \emph{predictive sensing} and proactive resource allocation. 
However, this tight coupling also introduces a fundamental challenge: inaccurate DT predictions may lead to suboptimal or even detrimental hardware configurations, which is different from the conventional fixed-array systems. 
Understanding and mitigating this closed-loop vulnerability is critical for reliable DT-driven operation.

\subsection{Related Work}

\textbf{CRB analysis for ISAC systems.}
The Cram\'{e}r--Rao bound (CRB) has become the standard sensing metric in ISAC system design.
Early works established the CRB--rate tradeoff for dual-function radar-communication waveforms~\cite{liu2020crb_dfrc}, and subsequent studies extended the analysis to MIMO-OFDM settings~\cite{liu2022isac}.
Recently, the fundamental capacity--CRB Pareto tradeoff in OFDM ISAC systems has been characterized in~\cite{huang2026barankin}, further tightening the connection between communication and sensing performance.
However, these analyses universally assume co-located arrays and focus on angle-of-arrival or delay estimation; \emph{velocity estimation from distributed Doppler measurements} and the impact of \emph{reconfigurable array morphology} on estimation bounds remain unexplored.

\textbf{Distributed sensing and velocity estimation.}
Distributed MIMO radar has long exploited multiple bistatic links for velocity resolution~\cite{li2008mimo_radar,haimovich2008mimo}, and several studies have analyzed the observability conditions and CRB for multi-static configurations~\cite{godrich2010target}.
In the cell-free ISAC context, existing works have primarily addressed localization accuracy through distributed angle or delay measurements~\cite{demirhan2024cellfree_isac,behdad2024multistatic}.
Joint position--velocity estimation, which combines angle-based localization with Doppler-based velocity reconstruction, has not been studied for cell-free systems with reconfigurable arrays.
The fundamental question of how many distributed links are required for full 2D velocity observability, and how reconfigurable element positions affect this observability, is unanswered.

\textbf{Digital twin for wireless networks.}
DT-assisted wireless systems have been applied to beam management~\cite{alkhateeb2023dt_beam}, channel estimation~\cite{jiang2023dt_beam_pred}, sensing-aided communication~\cite{Ding2024DTISAC}, and edge caching with reinforcement learning~\cite{Zhang2024DREC}.
A critical design question in DT architectures is \emph{synchronization}: how often and how much physical-layer data must be fed back to maintain twin fidelity. In~\cite{Yu2025DTSync} DT synchronization has been formulated as a joint resource-management problem but considered only communication-layer metrics, while a DT has been employed with an extended Kalman filter in ~\cite{Ding2024DTISAC} for vehicular tracking using the posterior CRB as a sensing constraint. Twin fidelity has been identified in ~\cite{bariah2024dt_frontier} as a first-class design metric, yet existing works do not characterize the closed-loop interplay between DT prediction quality, reconfigurable hardware adaptation, and sensing accuracy.
In particular, the question of when DT prediction errors cause reconfigurable hardware to \emph{actively degrade} performance and how to guarantee robustness through architectural choices has not been theoretically analyzed.

\subsection{Motivation and Scope}

The above discussion reveals a fundamental gap in the design of DT-driven cell-free ISAC systems: the lack of a unified framework that links \emph{estimation-theoretic limits} with \emph{closed-loop DT control}. 
This motivates the following questions:
%\begin{itemize}
%\item[\emph{Q1:}] What is the minimum number of distributed links required for full 2D velocity observability, and how does reconfigurable element positioning influence the estimation bounds?
%\item[\emph{Q2:}] Does FIM shape optimization improve both position and velocity estimation, or is there a fundamental decoupling between these two tasks?
%\item[\emph{Q3:}] Under what conditions do DT prediction errors cause optimized FIM configurations to degrade sensing performance, and how should synchronization policies adapt accordingly?
%\end{itemize}

%This motivates the following key questions:
\begin{itemize}
\item[\emph{Q1:}] What are the necessary and sufficient conditions for full 2D velocity observability in cell-free ISAC systems, and how does distributed link geometry govern the resulting estimation bounds?

\item[\emph{Q2:}] What is the fundamental relationship between spatial FIM shaping and temporal Doppler information in joint position--velocity estimation, and can these domains be optimized independently?

% \item[\emph{Q3:}] How do digital twin (DT) prediction errors impact the optimality of FIM-driven configurations, and what synchronization and control strategies are required to ensure robust performance?
\item[\emph{Q3:}] How do DT prediction errors interact with reconfigurable hardware optimization, and when does DT-driven adaptation become detrimental?
\item[\emph{Q4:}] What is the optimal DT synchronization rate and scheduling policy that balances FIM morphing gain against sensing overhead in a closed loop?
\end{itemize}

This paper addresses these issues by establishing a unified framework that bridges Fisher information analysis with DT-driven system design.

\subsection{Contributions}

The main contributions of this paper are as follows:
\begin{enumerate}
% \item \textbf{Velocity observability in cell-free ISAC (Q1, Theorems~\ref{thm:scatter}--\ref{thm:rank}).} 
% A scatter-matrix decomposition of the velocity Fisher information yields a waveform- and configuration-agnostic observability limit: single-BS systems are rank-$1$ for 2D velocity, while $M \geq 2$ spatially separated AP pairs are both necessary and sufficient for full-rank reconstruction, with RMSE scaling as $M^{-1/2}$ under fixed total power.

\item \textbf{Velocity observability in cell-free ISAC (Q1, Theorems~\ref{thm:scatter}--\ref{thm:rank}).} 
A scatter-matrix decomposition of the velocity Fisher information yields a waveform- and configuration-agnostic observability limit: single-BS systems are rank-$1$ for 2D velocity while two or more spatially separated access-point (AP) pairs are both necessary and sufficient for full-rank reconstruction, with the RMSE decreasing as the inverse square root of the number of APs under a fixed total power budget.

\item \textbf{Spatio-temporal orthogonality (Q2, Theorem~\ref{thm:ortho}, Propositions~\ref{prop:joint_fim}--\ref{prop:asym}).} 
From the $4{\times}4$ joint position--velocity CRB, we prove that FIM shape optimization strictly improves position accuracy but has no effect on velocity CRB under isotropic waveforms. The coupling is asymmetric: Doppler enhances position accuracy by up to $10.9$~dB at target speeds above $20$~m/s, however, position information has negligible effect on velocity bounds.

\item \textbf{Robustness of DT-driven adaptation (Q3, Theorem~\ref{thm:mismatch}).} 
We derive a closed-form critical mismatch angle %$\Delta\theta_{\mathrm{crit}} = |\hat{\theta}|$ 
beyond which DT-optimized FIM underperforms a fixed ULA, and prove that cell-free angular diversity preserves aggregate gain even when individual links invert, providing architectural robustness against DT prediction errors.

\item \textbf{Twin-in-the-loop synchronization (Q4, Propositions~\ref{prop:sync}--\ref{prop:horizon}, Algorithm~\ref{alg:policy}).} 
We introduce \emph{twin staleness} linking DT prediction age to FIM gain degradation, derive the optimal synchronization period %$T^{\star}$
in closed form, and identify a discrete phase transition of the prediction horizon at $M=2$. The resulting confidence-aware scheduler reduces the DT update rate by $2\times$ to $3\times$ over angle-only prediction. %A progressive ablation attributes the $23$~dB total gain over single-BS baselines to angular diversity ($10.2$~dB), Doppler--position fusion ($10.9$~dB), and FIM shape morphing ($1.9$~dB).
\end{enumerate}

\subsection{Paper Organization}
The remainder of this paper is organized as follows.
Section~\ref{sec:system} describes the FIM-augmented cell-free ISAC system model, including the signal model, FIM antenna architecture, and DT interface.
Section~\ref{sec:velocity_crb} derives the velocity CRB through scatter-matrix decomposition and establishes rank and scaling results.
Section~\ref{sec:joint_crb} develops the joint position--velocity FIM and proves the spatio-temporal orthogonality principle.
Section~\ref{sec:mismatch} analyzes FIM mismatch robustness under DT prediction errors.
Section~\ref{sec:dt_loop} presents the twin-in-the-loop closed-loop framework, including synchronization period optimization and adaptive scheduling.
Section~\ref{sec:results} provides extensive numerical validation.
Section~\ref{sec:conclusion} concludes the paper.

\emph{Notation:}
Boldface uppercase and lowercase letters denote matrices and vectors, respectively.
$(\cdot)^T$, $(\cdot)^H$, and $\tr(\cdot)$ denote transpose, Hermitian (conjugate transpose), and trace.
$\mat{I}_N$ is the $N \times N$ identity matrix, and $\|\cdot\|$ denotes the Euclidean norm.
$\Real\{\cdot\}$ takes the real part.
$\mathbb{E}[\cdot]$ denotes expectation, and $\succeq$ denotes positive semidefinite ordering.

%% file: 02_system_model.tex
\section{System Model}
\label{sec:system}
This section formalizes the FIM-augmented cell-free ISAC system and its 
digital-twin closed loop, and defines the joint position--velocity 
parameter space used throughout the paper.
\subsection{DT-Enabled System Architecture}
\label{sec:dt_arch}

\begin{figure}[t]
\centering
\includegraphics[width=\linewidth]{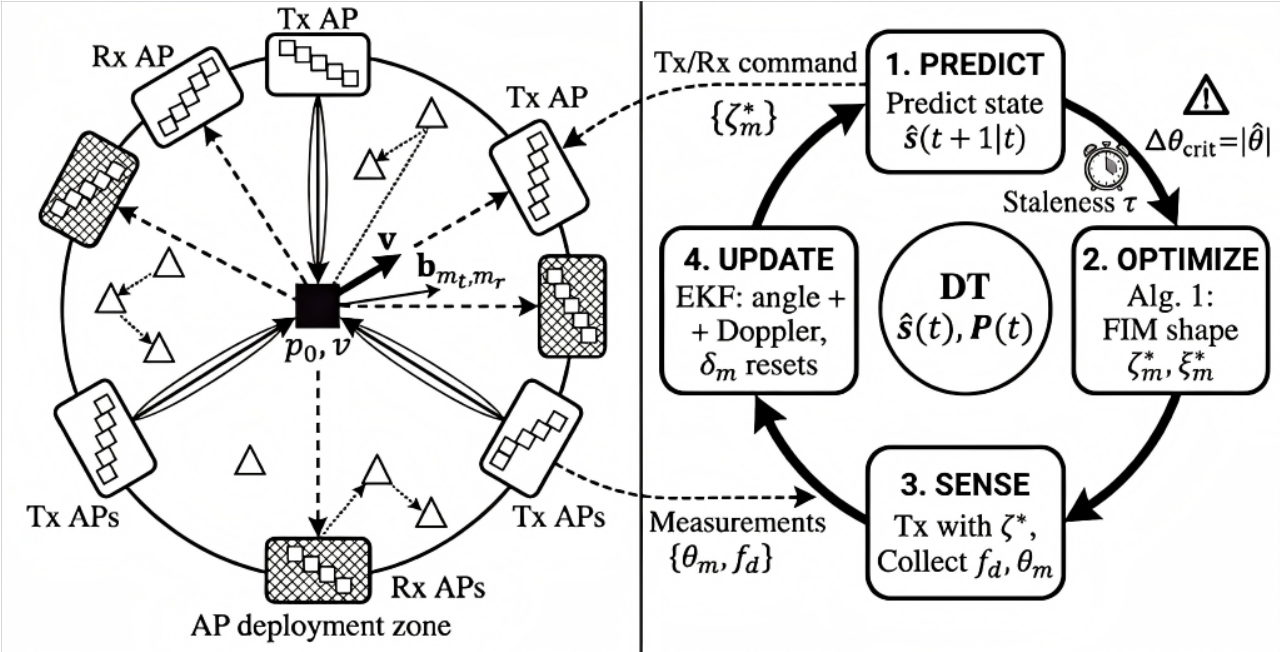}
\caption{FIM-augmented cell-free ISAC architecture}
\label{fig:sysmodel}
\end{figure}

We consider the cell-free ISAC system illustrated in Fig.~\ref{fig:sysmodel}. The physical layer comprises $M_t$ transmitting access points (Tx~APs) and $M_r$ receiving APs (Rx~APs) distributed over a geographical area and coordinated by a central processing unit (CPU). Each AP is equipped with a fluid intelligent metasurface (FIM) forming a uniform linear array of $N_t$ (respectively $N_r$) elements, whose element positions can be dynamically reconfigured within a bounded region, i.e., $|\zeta_{m_t,n}| \le \tilde{\zeta}$ for Tx~APs and $|\xi_{m_r,n}| \le \tilde{\zeta}$ for Rx~APs~\cite{zhu2024modeling}, where $\zeta_{m_t,n}$ (resp.\ $\xi_{m_r,n}$) denotes the displacement of the $n$-th element at Tx~AP~$m_t$ (resp.\ Rx~AP~$m_r$) from its nominal half-wavelength-spaced position. We collect the Tx and Rx displacement vectors as $\boldsymbol{\zeta}_{m_t} = [\zeta_{m_t,1},\ldots,\zeta_{m_t,N_t}]^\top \in \mathbb{R}^{N_t}$ and $\boldsymbol{\xi}_{m_r} = [\xi_{m_r,1},\ldots,\xi_{m_r,N_r}]^\top \in \mathbb{R}^{N_r}$, respectively. The APs jointly sense a point target at position $\mathbf{p}_0 = [x_0, y_0]^\top$ moving with velocity $\mathbf{v} = [v_x, v_y]^\top$ while simultaneously serving $K_u$ communication users, operating at carrier frequency $f_c$ with wavelength $\lambda = c/f_c$.

This model extends the angle-only framework of~\cite{he2026globecom} along three axes: expansion to joint position--velocity estimation ($\mathbb{R}^4$), inclusion of bistatic Doppler, and integration with a DT closed loop.

\subsubsection{Twin-in-the-Loop Operation}
As depicted in the right panel of Fig.~\ref{fig:sysmodel}, a digital twin (DT) module maintains a predictive state estimate and drives per-AP FIM reconfiguration through a closed loop. The DT holds a virtual replica of the target's kinematic state $\mathbf{s}_k = [\mathbf{p}_k^\top, \mathbf{v}_k^\top]^\top \in \mathbb{R}^4$ at each sensing epoch $k$, and operates in four stages:
\begin{enumerate}
  \item \textbf{Predict:} Given the posterior, $\hat{\mathbf{s}}_{k-1}$, the DT propagates the state through a kinematic model to obtain a predicted state $\bar{\mathbf{s}}_k$ and associated predicted angles $\{\bar{\theta}_m\}$, where $m$ indexes all Tx and Rx APs.
  \item \textbf{Optimize:} Using $\{\bar{\theta}_m\}$, the CPU solves a joint beamforming and FIM shape optimization problem to configure each AP's element displacements $\{\boldsymbol{\zeta}_{m_t}, \boldsymbol{\xi}_{m_r}\}$ for the upcoming coherent processing interval (CPI).
  \item \textbf{Sense:} The physical layer transmits ISAC waveforms and collects bistatic echoes across all Tx--Rx AP pairs, yielding angle and Doppler measurements $\mathbf{z}_k$.
  \item \textbf{Update:} An extended Kalman filter (EKF) fuses, $\mathbf{z}_k$, with the prior,  $\bar{\mathbf{s}}_k$ to produce the posterior,  $\hat{\mathbf{s}}_k$, which is fed back to the DT for the next cycle.
\end{enumerate}

\subsubsection{State-Space Formulation}
\label{sec:ssm}
The DT maintains a discrete-time state-space representation
\begin{align}
\mathbf{s}_k   &= \mathbf{F}\,\mathbf{s}_{k-1} + \mathbf{w}_{k-1},
  & \mathbf{w}_{k-1} &\sim \mathcal{N}(\mathbf{0}, \mathbf{Q}),
  \label{eq:process}\\
\mathbf{z}_k   &= h(\mathbf{s}_k) + \mathbf{v}_k,
  & \mathbf{v}_k &\sim \mathcal{N}(\mathbf{0}, \mathbf{R}_k),
  \label{eq:measurement}
\end{align}
where $\mathbf{F}$ is the constant-velocity state-transition matrix and $\mathbf{Q}$ the process-noise covariance obtained by sampling the continuous white-noise acceleration model at interval $\Delta t$:
% \begin{equation}
% \mathbf{F} = \begin{bmatrix} \mathbf{I}_2 & T\,\mathbf{I}_2 \\ \mathbf{0}_{2\times 2} & \mathbf{I}_2 \end{bmatrix},
% \quad
% \mathbf{Q} = \sigma_a^2\!\begin{bmatrix} \tfrac{T^4}{4}\mathbf{I}_2 & \tfrac{T^3}{2}\mathbf{I}_2 \\ \tfrac{T^3}{2}\mathbf{I}_2 & T^2\mathbf{I}_2 \end{bmatrix},
% \label{eq:F_Q}
% \end{equation}
\begin{equation}
\mathbf{F} = \begin{bmatrix} \mathbf{I}_2 & \Delta t\,\mathbf{I}_2 \\ \mathbf{0}_{2\times 2} & \mathbf{I}_2 \end{bmatrix},
\quad
\mathbf{Q} = q\begin{bmatrix} \tfrac{\Delta t^3}{3}\mathbf{I}_2 & \tfrac{\Delta t^2}{2}\mathbf{I}_2 \\ \tfrac{\Delta t^2}{2}\mathbf{I}_2 & \Delta t\,\mathbf{I}_2 \end{bmatrix},
\label{eq:F_Q}
\end{equation}
with $\Delta t$ the DT update interval, $q$ the process-noise power spectral density, and the measurement vector, $\mathbf{z}_k$, stacking all bearing angles $\{\theta_{m_t}, \theta_{m_r}\}$ and bistatic Doppler frequencies $\{f_{d,m_t,m_r}\}$, obtained from the array response and Doppler models defined later.
%derived from~\eqref{eq:array_response} and~\eqref{eq:doppler_freq}. 
% The measurement covariance, $\mathbf{R}_k$, is determined by the per-CPI CRB and thus depends on the FIM configuration $\{\boldsymbol{\zeta}_{m_t}, \boldsymbol{\xi}_{m_r}\}_k$, closing the loop between estimation bounds and DT update policy, a coupling analyzed in subsequent sections.
The measurement covariance, $\mathbf{R}_k$, is determined by the per-CPI CRB and thus depends on the FIM configuration $\{\boldsymbol{\zeta}_{m_t}, \boldsymbol{\xi}_{m_r}\}_k$, thereby closing the loop between estimation bounds and the DT update policy; this coupling is analyzed in subsequent sections.
%Sections~\ref{sec:mismatch} and~\ref{sec:dt_loop}.

\subsubsection{Assumptions}
%Throughout this paper, we adopt the following standard assumptions:
%\begin{itemize}\setlength\itemsep{1pt}\setlength\parsep{0pt}
%  \item[(A1)] Single point target in the far field.
%  \item[(A2)] AP positions stationary and known.
%  \item[(A3)] Stop-and-go target motion within each CPI.
%  \item[(A4)] Known path loss; RCS as unknown deterministic nuisance, jointly estimated with $\mathbf{s}_k$.
%  \item[(A5)] Isotropic $\mathbf{R}_{m_t}^{(\mathrm{s})}$ unless stated; required for Theorem~\ref{thm:ortho}.
%\end{itemize}

% Throughout this paper, we adopt the following standard assumptions:
% \begin{itemize}\setlength\itemsep{1pt}\setlength\parsep{0pt}
%   \item[(A1)] A single point target located in the far field.
%   \item[(A2)] AP positions are stationary and perfectly known.
%   \item[(A3)] Stop-and-go target motion within each CPI.
%   \item[(A4)] Path loss is known, while the radar cross section (RCS) is modeled as an unknown deterministic nuisance parameter jointly estimated with $\mathbf{s}_k$.
%   \item[(A5)] The transmit covariance, $\mathbf{R}_{m_t}^{(\mathrm{s})}$, is isotropic unless otherwise specified; this assumption is required for the orthogonality results developed later.
% \end{itemize}
Throughout this paper, we adopt the following standard assumptions:
\begin{list}{}{%
  \setlength{\leftmargin}{2.5em}
  \setlength{\labelwidth}{2em}
  \setlength{\labelsep}{0.7em}
  \setlength{\itemsep}{1.5pt}
  \setlength{\parsep}{0pt}
  \setlength{\topsep}{2.5pt}
}
  \item[(A1)] A single point target located in the far field.
  \item[(A2)] AP positions are stationary and perfectly known.
  \item[(A3)] Stop-and-go target motion within each CPI.
  \item[(A4)] Path loss is known, while the radar cross section (RCS) is modeled as an unknown deterministic nuisance parameter jointly estimated with $\mathbf{s}_k$.
  \item[(A5)] The transmit covariance, $\mathbf{R}_{m_t}^{(\mathrm{s})}$, is isotropic unless otherwise specified; this assumption is required for the orthogonality results developed later.
\end{list}
%----------------------------------------------------------------------
\subsection{Physical-Layer Signal Model}
\label{sec:signal}
%----------------------------------------------------------------------

The distributed AP deployment induces per-AP observation angles $\theta_{m_t} = \angle(\mathbf{p}_0 - \mathbf{p}_{m_t}^{(\mathrm{Tx})})$ and $\theta_{m_r} = \angle(\mathbf{p}_0 - \mathbf{p}_{m_r}^{(\mathrm{Rx})})$, where $\mathbf{p}_{m_t}^{(\mathrm{Tx})}$ and $\mathbf{p}_{m_r}^{(\mathrm{Rx})}$ are the known AP positions and $\angle(\cdot)$ is measured from the array broadside.

Each FIM-equipped AP has an array response parameterized by both the observation angle and the element-displacement vector. For Tx~AP~$m_t$ with half-wavelength nominal spacing, the $n$-th element response ($n = 1,\ldots,N_t$) is~\cite{zhu2024modeling}
\begin{equation}
[\mathbf{a}_{m_t}]_n = \frac{1}{\sqrt{N_t}}\exp\!\Big( j\big[(n\!-\!1)\pi\sin\theta_{m_t} + \tfrac{2\pi}{\lambda}\zeta_{m_t,n}\cos\theta_{m_t}\big]\Big),
\label{eq:array_response}
\end{equation}
where the second phase term captures FIM's physical shape deformation; setting $\boldsymbol{\zeta}_{m_t} = \mathbf{0}$ recovers the standard ULA. The Rx array response $\mathbf{a}_{m_r}$ follows the same structure with $m_t$, $N_t$, and $\boldsymbol{\zeta}_{m_t}$ replaced by $m_r$, $N_r$, and $\boldsymbol{\xi}_{m_r}$.

The angular sensitivity of the array is captured by the phase-rate vector
\begin{equation}
[\mathbf{d}_m]_n = (n\!-\!1)\pi\cos\theta_m - \tfrac{2\pi}{\lambda}\zeta_{m,n}\sin\theta_m,
\label{eq:phase_rate}
\end{equation}
so that $\dot{\mathbf{a}}_m \triangleq \partial\mathbf{a}_m/\partial\theta_m = j\,\mathrm{diag}(\mathbf{d}_m)\,\mathbf{a}_m$. Following~\cite{he2026globecom}, we define the per-AP effective aperture
\begin{equation}
\gamma_m(\boldsymbol{\zeta}_m) \triangleq \frac{1}{N_m}\|\mathbf{d}_m\|^2,\qquad N_m = \begin{cases}N_t, & m \in \mathcal{M}_t,\\ N_r, & m \in \mathcal{M}_r,\end{cases}
\label{eq:gamma}
\end{equation}
which governs the per-AP angle-estimation accuracy and is the key quantity enhanced by FIM morphing.\footnote{The definition~\eqref{eq:gamma} uses the raw phase-rate norm rather than the mean-centered $\|\mathbf{d}_m - \bar{d}_m\mathbf{1}\|^2$; consequently, a uniform element shift (constant $\zeta_{m,n} = c$) contributes to $\gamma_m$ through the $\sin\theta_m$ factor in~\eqref{eq:phase_rate}, consistent with the physical fact that the effect vanishes identically at broadside ($\theta_m = 0$).}

The bistatic sensing channel from Tx~AP~$m_t$ to Rx~AP~$m_r$ via the target is
\begin{equation}
\mathbf{G}_{m_t,m_r} = \alpha_{m_t,m_r}\,\mathbf{a}_{m_r}(\theta_{m_r},\boldsymbol{\xi}_{m_r})\,\mathbf{a}_{m_t}^H(\theta_{m_t},\boldsymbol{\zeta}_{m_t}),
\label{eq:channel}
\end{equation}
where $\alpha_{m_t,m_r} \in \mathbb{C}$ combines path loss and radar cross section (assumption~(A4)).
The received sensing signal at Rx~AP~$m_r$ aggregated over $L$ snapshots is $\mathbf{Y}_{m_r} = \sum_{m_t=1}^{M_t} \mathbf{G}_{m_t,m_r}\,\mathbf{X}_{m_t} + \mathbf{N}_{m_r},$ where $\mathbf{X}_{m_t} \in \mathbb{C}^{N_t \times L}$ is the transmitted waveform matrix and $\mathbf{N}_{m_r} \sim \mathcal{CN}(\mathbf{0}, \sigma^2\mathbf{I}_{N_r})$ is additive white noise. The waveform cross-correlation $\mathbf{R}_{m_0,m_1} \triangleq \mathbf{X}_{m_0}\mathbf{X}_{m_1}^H$ governs cross-AP information coupling in the sensing FIM~\cite{he2026globecom}.

%----------------------------------------------------------------------
\subsection{Communication Model}
\label{sec:comm}
%----------------------------------------------------------------------

The $M_t$ Tx~APs simultaneously serve $K_u$ single-antenna downlink users. Let $\mathbf{h}_{k,m_t} \in \mathbb{C}^{N_t}$ denote the channel from Tx~AP~$m_t$ to user~$k$, and $\mathbf{w}_{k,m_t} \in \mathbb{C}^{N_t}$ the corresponding beamforming vector, collected into $\mathbf{W}_{m_t} = [\mathbf{w}_{1,m_t}, \ldots, \mathbf{w}_{K_u,m_t}] \in \mathbb{C}^{N_t \times K_u}$.
The aggregate transmit signal at Tx~AP~$m_t$ comprises both communication and sensing components: $\mathbf{x}_{m_t} = \sum_{k=1}^{K_u}\mathbf{w}_{k,m_t}\,s_k + \mathbf{x}_{m_t}^{(\mathrm{s})},$ where $s_k$ is the data symbol for user~$k$ with $\mathbb{E}[|s_k|^2] = 1$, and $\mathbf{x}_{m_t}^{(\mathrm{s})}$ is the dedicated sensing waveform with covariance $\mathbf{R}_{m_t}^{(\mathrm{s})} \triangleq \mathbb{E}\big[\mathbf{x}_{m_t}^{(\mathrm{s})} \mathbf{x}_{m_t}^{(\mathrm{s})H}\big]$. The per-AP total transmit power is bounded by $\|\mathbf{W}_{m_t}\|_F^2 + \mathrm{tr}(\mathbf{R}_{m_t}^{(\mathrm{s})}) \leq P_{m_t}$.

The received signal at user~$k$ is
\begin{align}
y_k &= \sum_{m_t=1}^{M_t} \mathbf{h}_{k,m_t}^H \mathbf{w}_{k,m_t}\, s_k
     + \sum_{j \neq k} \sum_{m_t=1}^{M_t} \mathbf{h}_{k,m_t}^H \mathbf{w}_{j,m_t}\, s_j \notag\\
    &\quad + \sum_{m_t=1}^{M_t} \mathbf{h}_{k,m_t}^H \mathbf{x}_{m_t}^{(\mathrm{s})} + n_k,
\label{eq:rx_comm}
\end{align}
where $n_k \sim \mathcal{CN}(0, \sigma_c^2)$ is the receiver noise, and the three signal terms correspond to the desired signal, multi-user interference, and sensing waveform interference, respectively. Defining
$\mathsf{S}_k \triangleq \big|\textstyle\sum_{m_t}\mathbf{h}_{k,m_t}^H\mathbf{w}_{k,m_t}\big|^2$ and
$\mathsf{I}_k \triangleq \sum_{j\neq k}\big|\sum_{m_t}\mathbf{h}_{k,m_t}^H\mathbf{w}_{j,m_t}\big|^2 + \sum_{m_t}\mathbf{h}_{k,m_t}^H\mathbf{R}_{m_t}^{(\mathrm{s})}\mathbf{h}_{k,m_t}$,
the downlink SINR of user~$k$ is $\mathrm{SINR}_k = \mathsf{S}_k / (\mathsf{I}_k + \sigma_c^2)$.

The snapshot-level channel model is extended next to multiple snapshots to incorporate Doppler information from target mobility.

%----------------------------------------------------------------------
\subsection{Doppler Observation Model}
\label{sec:doppler}
%----------------------------------------------------------------------

Target mobility induces a time-varying phase across the $L$ snapshots within each CPI of duration $T_{\mathrm{cpi}} = L T_s$. For the bistatic pair $(m_t, m_r)$, the received signal at snapshot $l$ is
\begin{equation}
\mathbf{y}_l^{(m_t,m_r)} = \alpha_{m_t,m_r}\,e^{j2\pi f_{d,m_t,m_r}\, l T_s}\,\mathbf{a}_{m_r}\,\mathbf{a}_{m_t}^H\,\mathbf{x}_{m_t,l} + \mathbf{n}_l,
\label{eq:doppler_signal}
\end{equation}
where $l = 0,\ldots,L-1$, $T_s$ is the pulse repetition interval, $\mathbf{x}_{m_t,l} \in \mathbb{C}^{N_t}$ is the transmit signal at snapshot $l$ given by the $l$-th column of $\mathbf{X}_{m_t}$, and $\mathbf{n}_l \sim \mathcal{CN}(\mathbf{0}, \sigma^2 \mathbf{I}_{N_r})$. Equation~\eqref{eq:doppler_signal} thus makes explicit the inter-snapshot Doppler phase suppressed in the aggregate snapshot-level model.

The bistatic Doppler frequency encodes the target velocity through the \emph{bistatic bisector}:
\begin{equation}
f_{d,m_t,m_r} = \frac{1}{\lambda}\big(\mathbf{e}_{m_t} + \mathbf{e}_{m_r}\big)^\top \mathbf{v}
             = \frac{1}{\lambda}\,\mathbf{b}_{m_t,m_r}^\top \mathbf{v},
\label{eq:doppler_freq}
\end{equation}
where $\mathbf{e}_m = (\mathbf{p}_0 - \mathbf{p}_m)/\|\mathbf{p}_0 - \mathbf{p}_m\|$ is the unit line-of-sight vector and $\mathbf{b}_{m_t,m_r} \triangleq \mathbf{e}_{m_t} + \mathbf{e}_{m_r}$ is the bistatic bisector vector. In the monostatic special case $\mathbf{e}_{m_t} = \mathbf{e}_{m_r} = \mathbf{e}$, the bisector reduces to $\mathbf{b} = 2\mathbf{e}$, so that only the radial velocity component is observable; the geometric diversity of distinct bisectors across distributed AP pairs is what enables full 2D velocity reconstruction (cf.\ Theorem~\ref{thm:rank}).

\begin{remark}[Spatial vs.\ Temporal Information Channels]
\label{rmk:two_channels}
The signal model in~\eqref{eq:doppler_signal} features two information channels for sensing:
(i) the \emph{spatial channel}, where angular information is encoded in the array steering vectors, $\mathbf{a}_m(\theta_m, \boldsymbol{\zeta}_m)$, and is controllable by FIM shape $\boldsymbol{\zeta}_m$; and
(ii) the \emph{temporal channel}, where velocity information is encoded in the inter-snapshot phase progression, $e^{j2\pi f_d\, lT_s}$, and is determined solely by the bistatic geometry, $\{\mathbf{b}_{m_t,m_r}\}$.
Under isotropic sensing waveforms (assumption~(A5)), these channels are decoupled (Theorem~\ref{thm:ortho}): FIM morphing strictly improves the position CRB while having zero effect on the velocity CRB.
\end{remark}

%----------------------------------------------------------------------
\subsection{Estimation Parameter Space}
\label{sec:param}
%----------------------------------------------------------------------

By jointly exploiting angular and Doppler observations, the system estimates the target kinematic state, $\boldsymbol{\psi} = [\mathbf{p}_0^\top,\, \mathbf{v}^\top]^\top \in \mathbb{R}^4,$
where $\mathbf{p}_0$ and $\mathbf{v}$ denote the position and velocity.%respectively.

Geometrically, $\mathbf{p}_0$ is inferred via triangulation from the $M_t + M_r$ bearing angles while $\mathbf{v}$ is reconstructed from the $M_t M_r$ bistatic Doppler measurements in~\eqref{eq:doppler_freq}. Although these two modalities originate from distinct spatial and temporal channels, they are not fully independent: the Doppler frequency depends on the bistatic bisector $\mathbf{b}_{m_t,m_r}$, which is itself a function of the target position. Consequently, $\partial f_{d,m_t,m_r}/\partial \mathbf{p}_0 \neq \mathbf{0}$ introduces a cross-information coupling between position and velocity and plays a central role in the joint CRB analysis (Section~\ref{sec:joint_crb}).

For the parameter vector $\boldsymbol{\psi}$, we denote the Fisher information matrix by $\JFIM(\boldsymbol{\psi}) \in \mathbb{R}^{4 \times 4}$ and the Cram\'{e}r--Rao bound matrix by $\CRB(\boldsymbol{\psi}) = \JFIM^{-1}(\boldsymbol{\psi})$. The positional and velocity sub-CRBs are $\CRB_{\mathbf{p}} \triangleq \mathrm{tr}\!\big([\CRB(\boldsymbol{\psi})]_{1:2,1:2}\big)$ and $\CRB_{\mathbf{v}} \triangleq \mathrm{tr}\!\big([\CRB(\boldsymbol{\psi})]_{3:4,3:4}\big)$, respectively, with the corresponding RMSEs defined as $\sqrt{\CRB_{\mathbf{p}}}$ and $\sqrt{\CRB_{\mathbf{v}}}$.

%% file: 03_velocity_crb.tex
\section{Velocity CRB Analysis}
\label{sec:velocity_crb}

This section characterizes the fundamental limits of velocity estimation from Doppler observations.
We derive the velocity Fisher information matrix (FIM) and establish three key results:
(i) a scatter decomposition revealing the geometric structure of Doppler information;
(ii) a rank transition quantifying the capability gap between single-BS and cell-free architectures;
(iii) closed-form RMSE scaling laws describing how velocity accuracy improves with the number of APs.

%----------------------------------------------------------------------
\subsection{Scatter Decomposition of the Velocity FIM}
%----------------------------------------------------------------------

Applying the Slepian--Bangs formula to the Doppler signal model~\eqref{eq:doppler_signal} and aggregating over all bistatic pairs yields the velocity FIM (see Appendix~\ref{app:vel_fim}):
\begin{equation}
\FIM_{\vect{vv}}
= \frac{1}{\lambda^2}
\sum_{m_t=1}^{M_t} \sum_{m_r=1}^{M_r}
\eta^2_{m_t,m_r}\,
\vect{b}_{m_t,m_r}\vect{b}_{m_t,m_r}^\top,
\label{eq:vel_fim_total}
\end{equation}
where $\eta^2_{m_t,m_r} \triangleq \frac{8\pi^2 L}{3} T_\mathrm{cpi}^2 \cdot \mathrm{SNR}_{m_t,m_r}$ is the per-pair Doppler SNR factor, $T_\mathrm{cpi} = L T_s$ is the coherent processing interval, and $\vect{b}_{m_t,m_r} = \vect{e}_{m_t} + \vect{e}_{m_r}$ is the bistatic bisector.

\begin{theorem}[Scatter Decomposition]
\label{thm:scatter}
The velocity FIM admits the decomposition
\begin{equation}
\FIM_{\vect{vv}}
= \frac{1}{\lambda^2}
\bigl(
\mat{S}_{\mathrm{Tx}} + \mat{S}_{\mathrm{Rx}}
+ \mat{C} + \mat{C}^\top
\bigr),
\label{eq:scatter_decomp}
\end{equation}
where $\mat{S}_{\mathrm{Tx}} = \sum_{m_t,m_r} \eta^2_{m_t,m_r}\, \vect{e}_{m_t}\vect{e}_{m_t}^\top$, $\mat{S}_{\mathrm{Rx}} = \sum_{m_t,m_r} \eta^2_{m_t,m_r}\, \vect{e}_{m_r}\vect{e}_{m_r}^\top$, and $\mat{C} = \sum_{m_t,m_r} \eta^2_{m_t,m_r}\, \vect{e}_{m_t}\vect{e}_{m_r}^\top$.
\end{theorem}
\begin{IEEEproof}
Expanding $\vect{b}\vect{b}^\top = \vect{e}_{m_t}\vect{e}_{m_t}^\top + \vect{e}_{m_r}\vect{e}_{m_r}^\top + \vect{e}_{m_t}\vect{e}_{m_r}^\top + \vect{e}_{m_r}\vect{e}_{m_t}^\top$ in~\eqref{eq:vel_fim_total}, distributing the weighted sum yields~\eqref{eq:scatter_decomp}.
\end{IEEEproof}

The decomposition separates Doppler information into transmit-side, receive-side, and cross-coupling contributions. A key consequence appears in Theorem~\ref{thm:scaling}: under isotropic AP layouts, the cross-coupling term $\mat{C}$ vanishes, reducing the FIM to $\lambda^{-2}(\mat{S}_\mathrm{Tx}+\mat{S}_\mathrm{Rx})$ and enabling closed-form RMSE expressions. For APs uniformly distributed on a circle \emph{with target at the geometric center}, symmetry ensures $\kappa = 1$; for \emph{off-center targets} in our simulations with regular polygon layouts ($M = 3, 4, 5$), $\kappa$ ranges from 1.0 to 1.4 (cf.\ Section~\ref{sec:results}).

%----------------------------------------------------------------------
\subsection{Rank Transition: From 1D to 2D Velocity}
%----------------------------------------------------------------------

The velocity FIM in~\eqref{eq:vel_fim_total} is a sum of rank-1 matrices $\vect{b}_{m_t,m_r}\vect{b}_{m_t,m_r}^\top$, each contributing information only along its bisector direction.
Whether the full 2D velocity, $\vect{v}$, is identifiable therefore hinges on whether the collection of bisectors $\{\vect{b}_{m_t,m_r}\}$ spans $\mathbb{R}^2$, equivalently, whether $\FIM_{\vect{vv}}$ is full-rank.
This is the fundamental capability gap between single-BS and cell-free architectures.

\begin{theorem}[Rank Transition]
\label{thm:rank}
Consider the velocity FIM in~\eqref{eq:vel_fim_total} with $\eta^2_{m_t,m_r} > 0$ for all pairs.
\begin{enumerate}
\item[\emph{(i)}] \textbf{Single-BS:} If all APs are co-located ($\vect{e}_{m_t} = \vect{e}_{m_r} = \vect{e}$), then $\rank(\FIM_{\vect{vv}}) = 1$. Only the radial velocity is estimable.
\item[\emph{(ii)}] \textbf{Cell-free:} If at least two APs subtend distinct angles to the target, i.e., $\theta_m \neq \theta_{m'} \pmod{\pi}$ for some pair $(m, m')$, then $\rank(\FIM_{\vect{vv}}) = 2$. Both components of $\vect{v}$ are estimable.
\end{enumerate}
The isotropic assumption is not restrictive: for generic AP layouts, the same $M^{-1}$ and $M^{-1/2}$ scaling laws hold up to a constant condition-number factor $\kappa(\FIM_{\vect{vv}}) \in [1.2, 1.4]$ measured in our cell-free configurations (Fig.~\ref{fig:rank}, Remark~\ref{rem:cond}).
\end{theorem}

\begin{IEEEproof}
Since $\FIM_{\vect{vv}} = \lambda^{-2}\sum_{m_t,m_r}\eta^2_{m_t,m_r}\vect{b}_{m_t,m_r}\vect{b}_{m_t,m_r}^\top$ is a positive-weighted sum of rank-1 outer products, $\rank(\FIM_{\vect{vv}}) = \dim\mathrm{span}\{\vect{b}_{m_t,m_r}\}$. (i) When all APs are co-located, $\vect{b}_{m_t,m_r} = 2\vect{e}$ for all pairs, giving a 1D span.
(ii) Suppose the bisectors lie in a 1D subspace; then there exists a unit vector $\vect{u}$ orthogonal to this subspace with $\vect{u}^\top(\vect{e}_{m_t}+\vect{e}_{m_r}) = 0$ for all pairs. Fixing $m_r$, this forces $\vect{u}^\top\vect{e}_{m_t}$ to be constant across all Tx APs, implying that every Tx direction is parallel to a common line $\vect{u}^\perp$. Symmetrically for Rx. Contrapositively, if any two APs subtend distinct angles $\theta_m \neq \theta_{m'}\!\!\pmod{\pi}$, the bisector span is 2D.
\end{IEEEproof}

This theorem formalizes the \emph{capability gap} between co-located and distributed architectures: angular diversity directly translates into velocity observability.

\begin{remark}[Condition number and isotropy]
\label{rem:cond}
Rank-2 guarantees that both velocity components are estimable, but does not ensure they are estimated with equal accuracy.
The condition number, $\kappa(\FIM_{\vect{vv}}) \triangleq \lambda_{\max}/\lambda_{\min}$, quantifies this anisotropy: $\kappa = 1$ corresponds to isotropic estimation while $\kappa \to \infty$ recovers the single-BS case where tangential velocity is unobservable.
For APs uniformly distributed on a circle, symmetry ensures $\kappa = 1$; with regular polygon layouts in our simulation, $\kappa$ ranges from $1.0$ to $1.4$ (cf.\ Section~\ref{sec:results}), confirming near-isotropic estimation.
\end{remark}

%----------------------------------------------------------------------
\subsection{Scaling Laws}
%----------------------------------------------------------------------

From theorem~\ref{thm:rank}, cell-free architectures achieve full-rank velocity estimation, but does not quantify how the estimation accuracy improves as more APs are deployed.
We now derive closed-form RMSE scaling laws under idealized isotropic geometry, which serve as performance benchmarks for practical deployments.
We distinguish two practically relevant power regimes:
(A)~fixed per-link SNR, where each additional AP brings independent resources (e.g., dedicated backhaul power), and
(B)~fixed total power, where a constant power budget is split among all APs.

To enable closed-form analysis, we first establish a geometric identity for isotropic AP layouts.

% \begin{lemma}[Isotropic Geometry Identity]
% \label{lem:isotropic}
% Let $\{\vect{e}_m\}_{m=1}^{M}$ with $\vect{e}_m = [\cos\phi_m, \sin\phi_m]^\top$ and $\phi_m = 2\pi m/M$, $M \geq 3$. Then
% \begin{equation}
% \sum_{m=1}^{M} \vect{e}_m \vect{e}_m^\top = \frac{M}{2}\, \mat{I}_2.
% \end{equation}
% \end{lemma}
\begin{lemma}[Isotropic Geometry Identities]
\label{lem:isotropic}
Let $\vect{e}_m = [\cos\phi_m, \sin\phi_m]^\top$ with $\phi_m = 2\pi m/M$ and $M \geq 3$. Then $\sum_{m=1}^{M}\vect{e}_m = \vect{0}$ and $\sum_{m=1}^{M} \vect{e}_m\vect{e}_m^\top = \frac{M}{2}\, \mat{I}_2$.
\end{lemma}

\begin{IEEEproof}
The centroid identity $\sum_m \vect{e}_m = \vect{0}$ follows from $\sum_m e^{j\phi_m} = 0$ by geometric summation for $M \geq 2$. For the outer product, the diagonals satisfy $\sum_m\cos^2\phi_m = \sum_m\sin^2\phi_m = M/2$ via $\sum_m\cos 2\phi_m = \mathrm{Re}\{\sum_m e^{j4\pi m/M}\} = 0$ for $M \geq 3$, and the off-diagonal is $\sum_m\cos\phi_m\sin\phi_m = \tfrac{1}{2}\sum_m\sin 2\phi_m = 0$ by the same argument.
\end{IEEEproof}

% \begin{theorem}[Velocity RMSE Scaling Laws]
% \label{thm:scaling}
% Consider a symmetric cell-free deployment with $M_t = M_r = M$ where both Tx and Rx APs are isotropically deployed (Lemma~\ref{lem:isotropic}) and all per-pair SNR weights are equal within each power regime.
% \begin{enumerate}
% \item[\emph{(A)}] \textbf{Fixed per-link SNR} ($\eta^2_{m_t,m_r} = \eta_0^2$):
% \begin{equation}
% \mathrm{RMSE}_{\vect{v}}^{(A)}
% = \frac{\lambda\sqrt{2}}{M\,\eta_0}
% \;\propto\; M^{-1}.
% \label{eq:scaling_A}
% \end{equation}
% \item[\emph{(B)}] \textbf{Fixed total power} ($\eta^2_{m_t,m_r} = \eta_0^2/M$):
% \begin{equation}
% \mathrm{RMSE}_{\vect{v}}^{(B)}
% = \frac{\lambda\sqrt{2}}{\sqrt{M}\,\eta_0}
% \;\propto\; M^{-1/2}.
% \label{eq:scaling_B}
% \end{equation}
% \end{enumerate}
% \end{theorem}
\begin{theorem}[Velocity RMSE Scaling Laws]
\label{thm:scaling}
Consider a symmetric cell-free deployment with $M_t = M_r = M$, where both Tx and Rx APs are isotropically deployed (Lemma~\ref{lem:isotropic}) with equal per-pair SNR weights. Under \emph{(A)} fixed per-link SNR $\eta^2_{m_t,m_r} = \eta_0^2$, or \emph{(B)} fixed total power $\eta^2_{m_t,m_r} = \eta_0^2/M$,
\begin{equation}
\mathrm{RMSE}_{\vect{v}} = 
\begin{cases}
\dfrac{\lambda\sqrt{2}}{M\,\eta_0} \;\propto\; M^{-1}, & \text{Case~A},\\[6pt]
\dfrac{\lambda\sqrt{2}}{\sqrt{M}\,\eta_0} \;\propto\; M^{-1/2}, & \text{Case~B}.
\end{cases}
\label{eq:scaling}
\end{equation}
\end{theorem}
\begin{IEEEproof}
By Lemma~\ref{lem:isotropic}, $\sum_m \vect{e}_m = \vect{0}$, makes the cross-scatter vanish: $\mat{C} = \eta_0^2\bigl(\sum_{m_t}\vect{e}_{m_t}\bigr)\bigl(\sum_{m_r}\vect{e}_{m_r}\bigr)^\top = \mat{0}$. Applying Lemma~\ref{lem:isotropic} to the scatter matrices yields
$\mat{S}_\mathrm{Tx} = M_r\,\eta_0^2\sum_{m_t}\vect{e}_{m_t}\vect{e}_{m_t}^\top = (M^2\eta_0^2/2)\,\mat{I}_2$
and $\mat{S}_\mathrm{Rx} = (M^2\eta_0^2/2)\,\mat{I}_2$ for Case~A.
For Case~B, replacing $\eta_0^2$ by $\eta_0^2/M$ gives
$\mat{S}_\mathrm{Tx} = \mat{S}_\mathrm{Rx} = (M\eta_0^2/2)\,\mat{I}_2$. Assembling via Theorem~\ref{thm:scatter}: $\FIM_{\vect{vv}} = c\,\mat{I}_2$ with $c^{(A)} = M^2\eta_0^2/\lambda^2$ and $c^{(B)} = M\eta_0^2/\lambda^2$.
The RMSE follows from $\mathrm{RMSE} = \sqrt{\tr(\FIM_{\vect{vv}}^{-1})} = \sqrt{2/c}$.
\end{IEEEproof}

The $M^{-1}$ versus $M^{-1/2}$ gap has a clear interpretation: under fixed per-link SNR, adding APs scales the total Fisher information as $M^2$ (both the number of Tx--Rx pairs and each pair's SNR are preserved); under fixed total power, the per-link SNR decays as $1/M$, so the effective information grows only as $M$, recovering the classical square-root law for power-split multistatic sensing. %Section~\ref{sec:joint_crb} extends this analysis to the joint position--velocity FIM, where the spatial (FIM shape) and temporal (Doppler) information channels are shown to be orthogonal.

%% file: 04_joint_crb.tex
\section{Joint Position--Velocity CRB}
\label{sec:joint_crb}
%======================================================================

Recalling the kinematic state $\boldsymbol{\psi} = [\mathbf{p}_0^\top, \mathbf{v}^\top]^\top \in \mathbb{R}^4$ from Section~II, the corresponding FIM is $4\times 4$.  %this has been def in sectionII review here
This section establishes three key results: 
(i) the block structure of the joint FIM, 
(ii) an asymmetric coupling between position and velocity estimation, and 
(iii) a spatio-temporal orthogonality principle showing that FIM shape optimization affects position but not velocity.

\subsection{Joint FIM Structure}

%The joint parameter vector is $\bm{\psi} = [\vect{p}_0^\top, \vect{v}^\top]^\top \in \mathbb{R}^4$.

\begin{proposition}[Joint $4\times 4$ FIM]
\label{prop:joint_fim}
The joint FIM has the $2\times 2$ block structure
\begin{equation}
\JFIM(\boldsymbol{\psi})
= \begin{bmatrix}
\FIM_{\vect{pp}} & \FIM_{\vect{pv}} \\[4pt]
\FIM_{\vect{pv}}^\top & \FIM_{\vect{vv}}
\end{bmatrix}
\in \mathbb{R}^{4\times 4},
\label{eq:joint_fim}
\end{equation}
where
\begin{align}
\FIM_{\vect{pp}}
&= \FIM_{\vect{pp}}^{(\mathrm{angle})} + \FIM_{\vect{pp}}^{(\mathrm{Doppler})},
\label{eq:Ipp} \\
\FIM_{\vect{pp}}^{(\mathrm{Doppler})}
&= \sum_{m_t,m_r} \eta^2_{m_t,m_r}\, \vect{g}_{m_t,m_r}\vect{g}_{m_t,m_r}^\top,
\label{eq:Ipp_Doppler} \\
\FIM_{\vect{vv}}
&= \frac{1}{\lambda^2}\sum_{m_t,m_r} \eta^2_{m_t,m_r}\,\vect{b}_{m_t,m_r}\vect{b}_{m_t,m_r}^\top,
\label{eq:Ivv} \\
\FIM_{\vect{pv}}
&= \frac{1}{\lambda}\sum_{m_t,m_r} \eta^2_{m_t,m_r}\,\vect{g}_{m_t,m_r}\vect{b}_{m_t,m_r}^\top,
\label{eq:Ipv}
\end{align}
with $\mat{P}_m^\perp \triangleq \mat{I}_2 - \vect{e}_m\vect{e}_m^\top$, $r_m \triangleq \|\vect{p}_0 - \vect{p}_m\|$, and the Doppler--position coupling vector
\begin{equation}
\vect{g}_{m_t,m_r}
= \frac{1}{\lambda}\left(
\frac{\mat{P}_{m_t}^\perp}{r_{m_t}}
+ \frac{\mat{P}_{m_r}^\perp}{r_{m_r}}
\right) \vect{v}.
\label{eq:g_def}
\end{equation}
Since $\vect{g}_{m_t,m_r} \propto \vect{v}$, the coupling vanishes identically for stationary targets: $\FIM_{\vect{pv}} = \FIM_{\vect{pp}}^{(\mathrm{Doppler})} = \mat{0}$ when $\vect{v} = \vect{0}$.
\end{proposition}

\begin{IEEEproof}
Apply the chain rule to $f_{d,m_t,m_r} = \frac{1}{\lambda}\vect{b}_{m_t,m_r}^\top\vect{v}$, noting $\partial f_d/\partial\vect{v} = \frac{1}{\lambda}\vect{b}$ and $\partial f_d/\partial\vect{p}_0 = \vect{g}$ via the unit vector Jacobian $\partial\vect{e}_m/\partial\vect{p}_0 = \mat{P}_m^\perp / r_m$. The per-pair Doppler FIM is $\eta^2 [\vect{g};\;\frac{1}{\lambda}\vect{b}][\vect{g};\;\frac{1}{\lambda}\vect{b}]^\top$; summing over all pairs and adding $\FIM_{\vect{pp}}^{(\mathrm{angle})}$ yields the result.
\end{IEEEproof}

Here, $\vect{b}_{m_t,m_r} = \vect{e}_{m_t} + \vect{e}_{m_r}$ is the bistatic bisector defined in Section~\ref{sec:velocity_crb}, and $\vect{g}_{m_t,m_r}$ captures the Doppler--position coupling arising from the position dependence of the unit direction vectors.

\textit{Insight:} The joint FIM decomposes into \emph{pure spatial} (angle), \emph{pure temporal} (velocity), and \emph{cross-domain} coupling terms, providing a unified view of position–velocity estimation.

\subsection{Asymmetric Coupling}
The cross-block $\FIM_{\vect{pv}}$ in~\eqref{eq:Ipv} gives rise to an asymmetric coupling between position and velocity estimation, which is quantified next.

\begin{proposition}[Asymmetric Coupling]
\label{prop:asym}
Define the angle-only position CRB $\mathrm{CRB}_{\vect{pp}}^{(\mathrm{ang})} \triangleq (\FIM_{\vect{pp}}^{(\mathrm{angle})})^{-1}$ and the marginal velocity CRB $\mathrm{CRB}_{\vect{vv}}^{(\mathrm{vel})} \triangleq \FIM_{\vect{vv}}^{-1}$. Let $\mathrm{CRB}_{\vect{pp}}^{(\mathrm{joint})}$ and $\mathrm{CRB}_{\vect{vv}}^{(\mathrm{joint})}$ denote the corresponding sub-blocks of $\JFIM^{-1}$. Then:
\begin{enumerate}
\item[\emph{(i)}] \textbf{Position benefits:} $\mathrm{CRB}_{\vect{pp}}^{(\mathrm{joint})} \preceq \mathrm{CRB}_{\vect{pp}}^{(\mathrm{ang})}$;
\item[\emph{(ii)}] \textbf{Velocity pays:} $\mathrm{CRB}_{\vect{vv}}^{(\mathrm{joint})} \succeq \mathrm{CRB}_{\vect{vv}}^{(\mathrm{vel})}$.
\end{enumerate}
Equality holds iff $\FIM_{\vect{pv}} = \mat{0}$ (i.e., $\vect{v} = \vect{0}$). The coupling ratio $\rho \triangleq \tr(\mathrm{CRB}_{\vect{vv}}^{(\mathrm{joint})} - \mathrm{CRB}_{\vect{vv}}^{(\mathrm{vel})}) / \tr(\mathrm{CRB}_{\vect{vv}}^{(\mathrm{vel})})$ satisfies $\rho < 0.08$ for $\|\vect{v}\| \leq 50$~m/s (cf.\ Section~\ref{sec:results}).
\end{proposition}

\begin{IEEEproof}
The Doppler observations alone contribute a rank-one-per-pair FIM: $\eta^2_{m_t,m_r} [\vect{g};\;\frac{1}{\lambda}\vect{b}][\vect{g};\;\frac{1}{\lambda}\vect{b}]^\top$. Summing over all pairs yields the $4\times 4$ Doppler-only FIM
$\mat{F}^{(\mathrm{D})} = \begin{bmatrix} \FIM_{\vect{pp}}^{(\mathrm{Doppler})} & \FIM_{\vect{pv}} \\ \FIM_{\vect{pv}}^\top & \FIM_{\vect{vv}} \end{bmatrix} \succeq \mat{0},$ which is PSD as a sum of PSD rank-one matrices. Its Schur complement $\Delta_{\vect{pp}} = \FIM_{\vect{pp}}^{(\mathrm{Doppler})} - \FIM_{\vect{pv}}\FIM_{\vect{vv}}^{-1}\FIM_{\vect{pv}}^\top \succeq \mat{0}$ by the Schur complement lemma in~\cite{boyd2004convex}. The effective position FIM is then $\widetilde{\FIM}_{\vect{pp}} = \FIM_{\vect{pp}}^{(\mathrm{angle})} + \Delta_{\vect{pp}} \succeq \FIM_{\vect{pp}}^{(\mathrm{angle})}$; inverting gives part~(i). For velocity, $\widetilde{\FIM}_{\vect{vv}} = \FIM_{\vect{vv}} - \FIM_{\vect{pv}}^\top\FIM_{\vect{pp}}^{-1}\FIM_{\vect{pv}} \preceq \FIM_{\vect{vv}}$; inverting gives part~(ii).
\end{IEEEproof} \begin{remark}
Doppler improves position accuracy by providing additional spatial information, at the cost of a small penalty in velocity estimation.
\end{remark}%\textit{Insight:} Doppler measurements provide  \emph{additional spatial information}, improving position accuracy, but the need to jointly estimate position introduces a small penalty in velocity estimation.

\subsection{Spatial-Temporal Orthogonality}

\begin{theorem}[Spatial-Temporal Orthogonality]
\label{thm:ortho}
Under matched-filter beamforming ($\vect{w}_{m_t} = \vect{a}_{m_t}^*/\|\vect{a}_{m_t}\|$) or, more generally, any beamforming strategy whose per-link transmit gain is independent of the FIM shape parameters, the velocity FIM~$\FIM_{\vect{vv}}$ is independent of $\bm{\zeta}_{m_t}$ and $\bm{\xi}_{m_r}$. Specifically, for any scalar element $\zeta_{m_t,n}$ of~$\bm{\zeta}_{m_t}$ and $\xi_{m_r,n}$ of~$\bm{\xi}_{m_r}$:
\begin{equation}
\frac{\partial \FIM_{\vect{vv}}}
{\partial \zeta_{m_t,n}} = \mat{0},
\qquad
\frac{\partial \FIM_{\vect{vv}}}
{\partial \xi_{m_r,n}} = \mat{0},
\label{eq:ortho_statement}
\end{equation}
for all $m_t, m_r, n$. Consequently, optimizing FIM shape improves position but has no effect on velocity.
\end{theorem}

\begin{IEEEproof}
From~\eqref{eq:Ivv}, $\FIM_{\vect{vv}}$ depends on the FIM shape only through (a)~the Doppler SNR factor $\eta^2_{m_t,m_r}$ and (b)~the bisectors $\vect{b}_{m_t,m_r}$. We show that neither depends on $\zeta_{m_t,n}$ or $\xi_{m_r,n}$. \emph{Part (a).} Under matched-filter beamforming, $|\vect{a}_{m_t}^H\vect{w}_{m_t}|^2 = \|\vect{a}_{m_t}\|^2 = 1$ by the array normalization in~\eqref{eq:array_response}. Hence $\mathrm{SNR}_{m_t,m_r} = |\alpha_{m_t,m_r}|^2/\sigma^2$ and $\eta^2_{m_t,m_r}$ depend only on path loss, RCS, and noise—not on $\boldsymbol{\zeta}_{m_t}, \boldsymbol{\xi}_{m_r}$. \emph{Part (b).} The bisector $\vect{b}_{m_t,m_r} = \vect{e}_{m_t} + \vect{e}_{m_r}$ is determined by AP--target macroscopic geometry $\{\vect{p}_m, \vect{p}_0\}$, independent of intra-array element displacements. Combining (a) and (b), $\partial\FIM_{\vect{vv}}/\partial\zeta_{m_t,n} = \mat{0}$ and $\partial\FIM_{\vect{vv}}/\partial\xi_{m_r,n} = \mat{0}$.
\end{IEEEproof}

\begin{remark}[Scope and limitations]
\label{rem:scope}
The orthogonality holds for the marginal velocity FIM~$\FIM_{\vect{vv}}$. The joint velocity CRB depends weakly on the shape through the Schur complement involving~$\FIM_{\vect{pp}}$, but since the coupling ratio~$\rho$ is small, this indirect effect is negligible (velocity CRB variation $<0.01$~dB based on our simulation; cf.\ Section~\ref{sec:results}). %A second indirect channel exists in joint ISAC systems: changing the FIM shape may cause the beamformer~\cite{he2026globecom} to reallocate transmit power to maintain SINR constraints, modifying~$\eta^2$. This is an artifact of shared power budgets rather than a geometric dependency; under fixed sensing power, orthogonality is exact.
A second indirect channel arises in joint ISAC systems with a shared transmit power budget: altering the FIM shape changes the effective sensing channel gain, which in turn prompts the beamformer~\cite{he2026globecom} to reallocate power between sensing and communication to meet the SINR constraint, thereby modifying the per-pair Doppler SNR factor~$\eta^2$. This coupling is an artifact of power sharing, not a geometric dependency, under a fixed sensing power budget, orthogonality is exact.
\end{remark}

%The orthogonality principle simplifies the DT's online optimization: the FIM shape subproblem in Algorithm~\ref{alg:policy} can target position accuracy without incurring any velocity penalty.

The orthogonality principle simplifies the DT’s online optimization, enabling the FIM shape design to improve position accuracy without affecting velocity estimation; the associated algorithmic framework is developed in a subsequent section.

%% file: 05_mismatch_robust.tex
%======================================================================
\section{Mismatch Analysis and Robust FIM Design}
\label{sec:mismatch}
%======================================================================

The preceding analysis assumes perfect alignment between the DT-predicted angle and the true target direction during FIM configuration. In practice, however, the DT provides only a \emph{prediction} $\hat{\theta}_m$ with bounded uncertainty. This section characterizes how such mismatch impacts sensing performance, establishes fundamental degradation laws, and develops a robust FIM design that explicitly accounts for DT uncertainty within the closed-loop architecture.

\subsection{Mismatch Model}

We first analyze a single AP (subscript $m$ omitted for clarity); the per-AP results are applied component-wise to each of the $M_t + M_r$ APs in Section~\ref{sec:robust}. %Let $\hat{\theta}$ denote the DT-predicted angle and $\theta = \hat{\theta} - \Delta\theta$ the true angle. 
Let $\hat{\theta}$ denote the angle predicted by the DT, and let the true angle be $\theta = \hat{\theta} - \Delta \theta$.
The FIM is configured using the boundary-optimal solution~\cite{he2026globecom}, which maximizes the effective aperture at the predicted direction:
\begin{equation}
\zeta_n^*(\hat{\theta}) = -\tilde{\zeta}\,\mathrm{sgn}(\sin 2\hat{\theta}), \quad n = 1,\ldots,N,
\label{eq:zeta_star}
\end{equation}

Given an actual angle $\theta$, the phase-rate is
\begin{equation}
d_n(\theta, \boldsymbol{\zeta})
= (n-1)\pi\cos\theta
- \frac{2\pi}{\lambda}\zeta_n \sin\theta,
\label{eq:mm_phase_rate}
\end{equation}
and the effective aperture is $\gamma(\theta, \boldsymbol{\zeta}) = \frac{1}{N}\|\vect{d}\|^2$, which for a conventional ULA ($\boldsymbol{\zeta} = \vect{0}$) reduces to
\begin{equation}
\gamma_\mathrm{ULA}(\theta) = \frac{\pi^2(N-1)(2N-1)}{6}\cos^2\theta.
\label{eq:gamma_ULA}
\end{equation}

\subsection{Mismatch-Induced Performance Degradation}

\begin{theorem}[Mismatch Effective Aperture and Critical Threshold]
\label{thm:mismatch}
% Let the FIM shape be optimized for $\hat{\theta} \in (-\pi/2, \pi/2)\setminus\{0\}$ via~\eqref{eq:zeta_star}, operating at true angle $\theta = \hat{\theta} - \Delta\theta$. 
Suppose the FIM shape is configured based on the DT-predicted angle, $\hat{\theta} \in (-\pi/2, \pi/2) \setminus \{0\}$, via the boundary-optimal rule~\eqref{eq:zeta_star} while the target is actually located at true angle $\theta$. Denoting the prediction error by $\Delta\theta \triangleq \hat{\theta} - \theta$, the mismatched effective aperture decomposes as
\begin{equation}
\tilde{\gamma}(\Delta\theta)
= \gamma_\mathrm{ULA}(\theta)
+ \Gamma_\times(\Delta\theta)
+ \Gamma_\zeta(\Delta\theta),
\label{eq:mm_gamma_expanded}
\end{equation}
where $\Gamma_\times = \frac{2\pi^2\tilde{\zeta}(N-1)}{\lambda}\,\mathrm{sgn}(\sin 2\hat{\theta})\,\cos\theta\,\sin\theta$ is the cross-term between ULA and morphing contributions, and $\Gamma_\zeta = \frac{4\pi^2\tilde{\zeta}^2}{\lambda^2}\sin^2\!\theta$ is the pure morphing contribution. Define the normalized gain $R(\Delta\theta) \triangleq \tilde{\gamma}(\Delta\theta)/\gamma_\mathrm{ULA}(\theta)$. Then $R$ is strictly decreasing on $[0, |\hat{\theta}|)$ from $R(0) > 1$ to $R(|\hat{\theta}|) = 1$, and the critical mismatch threshold at which FIM ceases to outperform ULA is $\Delta\theta_\mathrm{crit} = |\hat{\theta}|$, beyond which $R(\Delta\theta) < 1$ (FIM underperforms ULA).
\end{theorem}

\begin{IEEEproof}
Substituting~\eqref{eq:zeta_star} into~\eqref{eq:mm_phase_rate} gives $d_n = (n-1)\pi\cos\theta + (2\pi\tilde{\zeta}/\lambda)\,\mathrm{sgn}(\sin 2\hat{\theta})\,\sin\theta$. Squaring, summing over $n$, and dividing by $N$ using $\frac{1}{N}\sum_{n=1}^{N}(n-1) = (N-1)/2$ and $\frac{1}{N}\sum_{n=1}^{N}(n-1)^2 = (N-1)(2N-1)/6$ yields~\eqref{eq:mm_gamma_expanded}. For $\hat{\theta} \in (0, \pi/2)$ and $\Delta\theta \in [0, \hat{\theta})$, $\theta \in (0, \hat{\theta}]$ with $\sin\theta, \cos\theta > 0$, so~\eqref{eq:mm_gamma_expanded} and~\eqref{eq:gamma_ULA} give $R = 1 + c_1\tan\theta + c_2\tan^2\!\theta$ with $c_1, c_2 > 0$; since $\tan\theta$ is strictly increasing on $(0, \pi/2)$ and $\theta$ strictly decreases in $\Delta\theta$, $R$ is strictly decreasing in $\Delta\theta$. At $\Delta\theta = \hat{\theta}$, $\theta = 0$ forces $\Gamma_\times = \Gamma_\zeta = 0$ (both carry a $\sin\theta$ factor), hence $R(\hat{\theta}) = 1$. For $\Delta\theta > \hat{\theta}$, $\theta < 0$ flips the sign of $\Gamma_\times$ while $\Gamma_\zeta \geq 0$, driving $R < 1$. The case $\hat{\theta} \in (-\pi/2, 0)$ follows by symmetry $\hat{\theta} \to -\hat{\theta}$.
\end{IEEEproof}

\begin{remark}[Broadside blind spot]
\label{rem:sin_zero}
When $\theta = 0$, the morphing term $-(2\pi/\lambda)\zeta_n\sin 0$ vanishes identically: the wavefront is parallel to the array axis and element displacement produces no path-length difference. Consequently, no choice of $\boldsymbol{\zeta}$ can improve angle resolution over a ULA at broadside, which explains why mismatch is most damaging when the true angle crosses broadside while the design angle is off-broadside. This also motivates excluding $\hat{\theta} = 0$ in Theorem~\ref{thm:mismatch}: at exactly broadside, FIM and ULA coincide and $R \equiv 1$ identically.
\end{remark}

\subsection{Robust FIM Shape Optimization}
\label{sec:robust}

Applying Theorem~\ref{thm:mismatch} to each AP $m$, the mismatch analysis motivates a robust design hedging against DT uncertainty $\theta_m \in [\hat{\theta}_m - \delta_m,\, \hat{\theta}_m + \delta_m]$. By Theorem~\ref{thm:ortho}, FIM shape affects only the position CRB; the velocity CRB is shape-independent. Hence, the robust minimax is formulated over the position CRB alone:
\begin{subequations}\label{eq:robust_P0}
\begin{align}
\min_{\{\boldsymbol{\zeta}_m, \boldsymbol{\xi}_m\}, \{\mat{W}_{m_t}\}} \;
&\max_{\boldsymbol{\theta} \in \Theta_\delta} \;\;
\mathrm{CRB}_{\vect{p}}\!\big(\boldsymbol{\theta}; \{\boldsymbol{\zeta}_m, \boldsymbol{\xi}_m, \mat{W}_{m_t}\}\big)
\label{eq:robust_obj} \\
\text{s.t.} \;\;
&\mathrm{SINR}_k \geq \Gamma_k,\;
\|\mat{W}_{m_t}\|_F^2 \leq P_{m_t},\;
|\zeta_{m,n}| \leq \tilde{\zeta}.
\label{eq:robust_constraints}
\end{align}
\end{subequations}
A second-order Taylor expansion of the per-AP position CRB contribution around $\hat{\theta}_m$ gives
\begin{equation}
\mathrm{CRB}_{\vect{p},m}(\theta_m) \approx h_{0,m}(\boldsymbol{\zeta}_m) + h_{1,m}(\boldsymbol{\zeta}_m)\,\Delta\theta_m + h_{2,m}(\boldsymbol{\zeta}_m)\,\Delta\theta_m^2,
\label{eq:taylor_crb}
\end{equation}
where $\Delta\theta_m = \theta_m - \hat{\theta}_m$ and $h_{0,m}, h_{1,m}, h_{2,m}$ depend on the FIM shape $\boldsymbol{\zeta}_m$, but not on the uncertain angle. The inner maximization over $|\Delta\theta_m| \leq \delta_m$ is a scalar Quadratically Constrained Quadratic Program (QCQP). Introducing an auxiliary epigraph variable $t_m$ and applying the S-procedure~\cite{boyd2004convex} (see Appendix~\ref{app:sdp}), we convert the semi-infinite constraint into a finite Semidefinite Programming (SDP):
\begin{equation}
\min_{\{\boldsymbol{\zeta}_m\},\, \tau_m \geq 0,\, t_m}\;
\sum_m t_m
\quad\text{s.t.}\;\eqref{eq:robust_lmi},\;
|\zeta_{m,n}| \leq \tilde{\zeta},
\label{eq:robust_sdp}
\end{equation}
whose optimizer, denoted $\boldsymbol{\chi}_{\mathrm{rob}}(\hat{\theta}_m, \delta_m)$, is the robust FIM shape used in Algorithm~\ref{alg:policy}.
where the $2\times 2$ LMI constraint for each AP $m$ is
\begin{equation}
\begin{bmatrix}
t_m - h_{0,m}(\boldsymbol{\zeta}_m) - \tau_m\delta_m^2 & -h_{1,m}(\boldsymbol{\zeta}_m)/2 \\[3pt]
-h_{1,m}(\boldsymbol{\zeta}_m)/2 & \tau_m - h_{2,m}(\boldsymbol{\zeta}_m)
\end{bmatrix}
\succeq \mat{0}.
\label{eq:robust_lmi}
\end{equation}
Here $\boldsymbol{\zeta}_m$ enters through coefficients $h_{0,m}, h_{1,m}, h_{2,m}$, while uncertainty $\Delta\theta_m$ has been eliminated by the S-procedure. The auxiliary, $t_m$, acts as the robust upper bound for the position CRB at AP $m$, and the $2\times 2$ LMI structure reflects the scalar nature of the per-AP angular uncertainty. A key qualitative difference from the nominal design: the robust optimum is an interior point ($|\zeta_{m,n}^*| < \tilde{\zeta}$), sacrificing peak gain for broader angular coverage. As $\delta_m \to 0$, $\boldsymbol{\zeta}_m^{*,\mathrm{rob}} \to \boldsymbol{\zeta}_m^*$ (nominal boundary); as $\delta_m \to |\hat{\theta}_m|$, $\boldsymbol{\zeta}_m^{*,\mathrm{rob}} \to \vect{0}$ (ULA).

\begin{remark}[DT fidelity--performance trade-off]
\label{rem:dt_rate}
The achievable FIM gain $G_m(\delta_m)$ at AP $m$ decreases monotonically with DT uncertainty $\delta_m$, establishing a direct link between DT prediction accuracy and sensing performance. Ensuring $\delta_m < |\hat{\theta}_m|$ is necessary to guarantee positive morphing gain at AP $m$, providing a quantitative guideline for the DT synchronization and update policies developed in Section~\ref{sec:dt_loop}.
\end{remark}

Section~\ref{sec:dt_loop} uses $G_m(\delta_m)$ as the key input to derive the optimal DT synchronization period, $T^\star$, in closed form, balancing morphing gain against sensing overhead.

%% file: 06_dt_loop.tex
\section{Digital Twin Closed-Loop Analysis}
\label{sec:dt_loop}
%======================================================================

The preceding sections characterized CRB performance under fixed and perfectly known target parameters. In practice, however, FIM shape optimization relies on the DT's predicted target state, which inevitably degrades between sensing updates. This section develops a DT-in-the-loop framework that captures this dynamic interplay. Specifically, we present the EKF-based state estimator (Section~\ref{sec:ekf}), introduce a twin staleness model linking DT age to FIM degradation (Section~\ref{sec:staleness}), propose a confidence-aware adaptive FIM policy (Section~\ref{sec:policy}), derive the optimal DT synchronization period (Section~\ref{sec:sync}), analyze the prediction horizon (Section~\ref{sec:horizon}), and quantify the benefits of Doppler-enhanced tracking.

% \subsection{EKF State Estimation}
% \label{sec:ekf}

% The EKF state vector is $\mathbf{s} = [\mathbf{p}_0^\top, \mathbf{v}^\top]^\top \in \mathbb{R}^4$ with the constant-velocity transition model $\mathbf{s}(t+1) = \mathbf{F}\mathbf{s}(t) + \mathbf{w}(t)$, where
% \begin{equation}
% \mathbf{F} = \begin{bmatrix} \mathbf{I}_2 & \Delta t\, \mathbf{I}_2 \\ \mathbf{0} & \mathbf{I}_2 \end{bmatrix}, \quad
% \mathbf{Q} = q \begin{bmatrix} \frac{\Delta t^3}{3}\mathbf{I}_2 & \frac{\Delta t^2}{2}\mathbf{I}_2 \\ \frac{\Delta t^2}{2}\mathbf{I}_2 & \Delta t\, \mathbf{I}_2 \end{bmatrix},
% \label{eq:ekf_model}
% \end{equation}
% with $\Delta t$ the update interval and $q$ the process noise power spectral density. The measurement vector comprises angle observations $\{\theta^{(\mathrm{AoD})}_{m_t}, \theta^{(\mathrm{AoA})}_{m_r}\}$ and Doppler observations $\{f_{d,m_t,m_r}\}$, with measurement noise covariances determined by the respective CRBs. The EKF update step produces the posterior state estimate $\hat{\mathbf{s}}(t|t)$ and covariance $\mathbf{P}(t|t)$; the next-epoch prior covariance $\mathbf{P}(t+1|t) = \mathbf{F}\mathbf{P}(t|t)\mathbf{F}^\top + \mathbf{Q}$ is propagated through the CV model.

\subsection{EKF State Estimation}
\label{sec:ekf}

The EKF state vector is $\mathbf{s} = [\mathbf{p}_0^\top, \mathbf{v}^\top]^\top \in \mathbb{R}^4$, propagated by the constant-velocity transition model, $\mathbf{s}_{k} = \mathbf{F}\mathbf{s}_{k-1} + \mathbf{w}_{k-1}$, defined in~\eqref{eq:F_Q}. The measurement vector comprises angle observations $\{\theta^{(\mathrm{AoD})}_{m_t}, \theta^{(\mathrm{AoA})}_{m_r}\}$ and Doppler observations $\{f_{d,m_t,m_r}\}$, with measurement noise covariances determined by the respective CRBs. The EKF update step produces the posterior state estimate, $\hat{\mathbf{s}}(t|t)$, and covariance, $\mathbf{P}(t|t)$; the next-epoch prior, $\mathbf{P}(t+1|t) = \mathbf{F}\mathbf{P}(t|t)\mathbf{F}^\top + \mathbf{Q}$, is propagated through the constant-velocity model.

Since FIM shape decisions are made before the next sensing epoch, we extract the \emph{prior} angular uncertainty at each AP $m$ from $\mathbf{P}(t+1|t)$,
\begin{equation}
\delta_m(t) = \sqrt{[\mathbf{J}_m \mathbf{P}(t+1|t) \mathbf{J}_m^\top]_{1,1}},
\label{eq:delta_from_P}
\end{equation}
where $\mathbf{J}_m = \partial \theta_m / \partial \mathbf{p}_0$ is the angle Jacobian. This quantity feeds directly into the mismatch analysis of Theorem~\ref{thm:mismatch}.

\begin{remark}[Posterior CRB]
\label{rem:pcrb}
The DT's prior covariance, $\mathbf{P}(t+1|t)$, augments the data FIM, 
yielding the posterior Fisher information, $\mathbf{J}_{\mathrm{post}}(t) 
= \mathbf{P}(t+1|t)^{-1} + \mathbf{J}_{\mathrm{data}}(t)$, and posterior 
CRB~\cite{tichavsky1998pcrb}, $\mathrm{PCRB}(t) = \mathrm{tr}(\mathbf{J}_{
\mathrm{post}}^{-1}) \leq \mathrm{CRB}(t)$, with equality iff the DT 
carries no prior. The gap is largest for velocity under single-BS, 
where the data FIM is rank-deficient (Theorem~\ref{thm:rank}) while 
the temporal prior remains full-rank.
\end{remark}

\subsection{Twin Staleness and FIM Gain Degradation}
\label{sec:staleness}

\begin{definition}[Twin Staleness]
\label{def:staleness}
Let $t_0$ denote the last DT update time. The twin staleness, $\tau$, is the elapsed time since that update. Over $[t_0, t_0+\tau]$, the DT operates purely in prediction mode without measurement corrections.
\end{definition}

As $\tau$ increases, prediction errors accumulate, causing the predicted angle, $\bar{\theta}_m$, to drift from the true angle, $\theta_m$. Under the constant-velocity model, this drift is governed by velocity estimation error rather than absolute target speed. The mismatch evolves approximately as
\begin{equation}
\Delta\theta_m(\tau) \approx \omega_m \cdot \tau, \quad \omega_m \triangleq \frac{\sigma_{v}^{(\mathrm{mode})}}{r_m},
\label{eq:angular_drift}
\end{equation}
%where $\sigma_{v}^{(\mathrm{mode})}$ denotes the perpendicular component of the DT's steady-state velocity estimation error under the operating mode (approximated as the total RMSE under isotropic velocity error). For an angle-only EKF, velocity must be inferred indirectly from position changes across time steps, 
%yielding $\sigma_v^{(\mathrm{angle})} \approx 0.6$--$1.1$~m/s depending on target speed; incorporating Doppler provides direct velocity observations, reducing the residual to $\sigma_v^{(\mathrm{Dopp})} \approx 0.15$~m/s (cf.\ Fig.~\ref{fig:dt}).
where $\sigma_{v}^{(\mathrm{mode})}$ denotes the perpendicular component of the DT's steady-state velocity estimation error under the operating mode (approximated as the total RMSE under isotropic velocity error). For an angle-only EKF, velocity must be inferred indirectly from position changes across time steps, producing a nontrivial residual $\sigma_v^{(\mathrm{angle})}$; incorporating Doppler provides direct velocity observations and reduces this residual by roughly an order of magnitude to $\sigma_v^{(\mathrm{Dopp})}$ (cf.\ Fig.~\ref{fig:dt}).

By Theorem~\ref{thm:mismatch}, the per-link FIM gain, $R(\Delta\theta)$, is monotonically decreasing on $[0, |\hat\theta|)$ and crosses the ULA baseline at $\Delta\theta_\mathrm{crit} = |\hat{\theta}|$. The system-level FIM gain, aggregated over all AP links, therefore degrades over the staleness interval:
\begin{equation}
G(\tau) = \frac{\sum_{m} w_m\, \tilde{\gamma}_m\bigl(\Delta\theta_m(\tau)\bigr)}{\sum_{m} w_m\, \gamma_{\mathrm{ULA},m}\bigl(\theta_m - \Delta\theta_m(\tau)\bigr)},
\label{eq:G_tau}
\end{equation}
where $w_m$ weights each link by its SNR contribution. This transforms the static mismatch result of Section~\ref{sec:mismatch} into a time-varying DT fidelity metric.

\subsection{Confidence-Aware Adaptive FIM Policy}
\label{sec:policy}

The theoretical results of Sections~\ref{sec:velocity_crb}--\ref{sec:mismatch} converge into a unified operational policy for DT-driven FIM reconfiguration. At each sensing epoch~$k$, the DT provides a predicted angle $\bar{\theta}_{m,k}$ for each AP~$m$ along with an angular uncertainty $\delta_{m,k}$ derived from the prior covariance via~\eqref{eq:delta_from_P}. The closed-form critical mismatch $\Delta\theta_{\mathrm{crit}} = |\hat{\theta}|$ (Theorem~\ref{thm:mismatch}) directly partitions the uncertainty space into three operating regimes, eliminating the need for iterative robustness optimization at run time.

Algorithm~\ref{alg:policy} formalizes this policy. For each AP, the DT compares $\delta_{m,k}$ against two thresholds:
\begin{itemize}
\item \emph{Confident regime} ($\delta_{m,k} < \delta_{\mathrm{low}}$): nominal FIM optimization (boundary solution~\eqref{eq:zeta_star}) captures full morphing gain with negligible mismatch risk.
\item \emph{Uncertain regime} ($\delta_{\mathrm{low}} \le \delta_{m,k} < |\bar{\theta}_{m,k}|$): mismatch is non-negligible but recoverable. The robust interior shape from the LMI formulation~\eqref{eq:robust_sdp} hedges against worst-case angle deviation (Remark~\ref{rem:dt_rate}).
\item \emph{Degraded regime} 
($\delta_{m,k} \ge |\bar{\theta}_{m,k}|$): 
the uncertainty exceeds the critical mismatch (Theorem~\ref{thm:mismatch}), 
so the policy reverts to a fixed ULA ($\boldsymbol{\zeta}_m = \mathbf{0}$) 
as the safe fallback.
\end{itemize}
After per-AP FIM shapes are determined, the joint beamforming problem is solved via SDR (as in~\cite{he2026globecom}) with the updated array responses. Because $\delta_m$ is computed per AP, different links may operate in different regimes simultaneously, where an AP with accurate DT prediction runs the confident mode while a poorly predicted link falls back to ULA, yielding robustness through architectural diversity.

\begin{algorithm}[t]
\caption{Confidence-Aware Adaptive FIM Policy}
\label{alg:policy}
\begin{algorithmic}[1]
\small
%\Require DT posterior $\hat{\mathbf{s}}_{k|k}$, covariance $\mathbf{P}_{k|k}$, threshold $\delta_{\mathrm{low}}$
\Require DT posterior $\hat{\mathbf{s}}_{k|k}$, covariance $\mathbf{P}_{k|k}$, confidence threshold $\delta_{\mathrm{low}}$, pre-computed robust shape $\boldsymbol{\chi}_{\mathrm{rob}}(\cdot,\cdot)$ from~\eqref{eq:robust_sdp}
\Ensure FIM shapes $\{\boldsymbol{\zeta}_{m_t}^*, \boldsymbol{\xi}_{m_r}^*\}$, beamformers $\{\mathbf{W}_{m_t}^*\}$
\State \textbf{Predict:} Propagate to obtain $\bar{\mathbf{s}}_{k+1|k}$, $\mathbf{P}_{k+1|k}$
\State Compute predicted angles $\{\bar{\theta}_m\}$ from $\bar{\mathbf{s}}_{k+1|k}$
\For{each AP $m \in \mathcal{M}_t \cup \mathcal{M}_r$}
  \State $\delta_m \gets \sqrt{[\mathbf{J}_m \mathbf{P}_{k+1|k} \mathbf{J}_m^\top]_{1,1}}$
  \State Let $\boldsymbol{\chi}_m$ denote $\boldsymbol{\zeta}_{m}$ if $m \in \mathcal{M}_t$, else $\boldsymbol{\xi}_{m}$
  \If{$\delta_m < \delta_{\mathrm{low}}$} \Comment{Confident}
    \State $\boldsymbol{\chi}_m^* \gets -\tilde{\zeta}\, \mathrm{sgn}(\sin 2\bar{\theta}_m)\,\mathbf{1}$ \Comment{Eq.~\eqref{eq:zeta_star}}
  \ElsIf{$\delta_m < |\bar{\theta}_m|$} \Comment{Uncertain}
    \State $\boldsymbol{\chi}_m^* \gets \boldsymbol{\chi}_{\mathrm{rob}}(\bar{\theta}_m, \delta_m)$ \Comment{Eq.~\eqref{eq:robust_sdp}}
  \Else \Comment{Degraded: ULA fallback}
    \State $\boldsymbol{\chi}_m^* \gets \mathbf{0}$
  \EndIf
\EndFor
\State \textbf{Beamform:} Solve SDR with $\{\boldsymbol{\chi}_m^*\}$ to obtain $\{\mathbf{W}_{m_t}^*\}$
\State \textbf{Sense:} Transmit, collect echoes, extract $\{\hat{\theta}_m, \hat{f}_d\}$
\State \textbf{Update:} EKF measurement update $\to$ $\hat{\mathbf{s}}_{k+1|k+1}$, $\mathbf{P}_{k+1|k+1}$
\end{algorithmic}
\end{algorithm}

\begin{remark}[Theory-Enabled Simplicity]
\label{rmk:simplicity}
The closed-form critical mismatch, $\Delta\theta_{\mathrm{crit}} = 
|\hat{\theta}|$ and monotonicity of $G_m(\delta_m)$, 
(Theorem~\ref{thm:mismatch}) reduce the per-AP decision in 
Algorithm~\ref{alg:policy} to threshold comparisons. Conventional 
robust beamforming requires solving an SDP at every epoch with 
per-iteration cost, $\mathcal{O}(N^{3.5})$; here, $\boldsymbol{\chi}_{
\mathrm{rob}}$, is pre-computed offline on a grid of 
$(\bar{\theta}, \delta)$ via the LMI~\eqref{eq:robust_sdp}.
% The per-AP branching in Algorithm~\ref{alg:policy} is intentionally simple, a direct consequence of the closed-form critical mismatch $\Delta\theta_{\mathrm{crit}} = |\hat{\theta}|$ and the monotonicity of $G_m(\delta_m)$ established in Theorem~\ref{thm:mismatch}. Conventional robust beamforming would require solving an SDP at every sensing epoch with per-iteration cost $\mathcal{O}(N^{3.5})$; here, the online decision reduces to threshold comparisons, with the robust shape, $\boldsymbol{\chi}_{\mathrm{rob}}$, pre-computed offline on a grid of $(\bar{\theta}, \delta)$ values via the LMI~\eqref{eq:robust_sdp}.
\end{remark}

\begin{remark}[Computational Complexity]
\label{rmk:complexity}
The online complexity of Algorithm~\ref{alg:policy} decomposes as (i) $\mathcal{O}(M)$ matrix--vector multiplications for $\delta_m$ computation; (ii) $\mathcal{O}(M)$ threshold comparisons with $\mathcal{O}(1)$ lookup per AP for the robust shape; (iii) $\mathcal{O}((M_t N_t)^{3.5})$ for SDR beamforming, identical to the nominal design in~\cite{he2026globecom}. Steps (i)--(ii) add only $\mathcal{O}(M)$ overhead. For the parameters in Table~\ref{tab:params} ($M = 4$, $N_t = 8$), FIM selection requires $< 0.1$~ms on a standard CPU versus $\sim 50$~ms for SDR beamforming, confirming that policy overhead is negligible.
\end{remark}

\begin{table}[t]
\centering
\caption{Online computational complexity. The proposed policy matches the nominal design cost $\mathcal{O}(M)$ while providing worst-case robustness via Theorem~\ref{thm:mismatch}.}
\label{tab:complexity}
\renewcommand{\arraystretch}{1.15}
\begin{tabular}{lcc}
\toprule
\textbf{Strategy} & \textbf{FIM Selection Cost} & \textbf{Guarantee} \\
\midrule
Nominal (no robustness) & $\mathcal{O}(M)$ & None \\
Online minimax robust & $\mathcal{O}(M N^{3.5})$ & Worst-case \\
MPC ($H$-step horizon) & $\mathcal{O}(H \cdot M N^{3.5})$ & Predictive \\
\textbf{Proposed policy} & $\mathcal{O}(M)$ & \textbf{Thm.~\ref{thm:mismatch}} \\
\bottomrule
\end{tabular}
\end{table}

\subsection{Optimal DT Synchronization Period}
\label{sec:sync}

The DT update entails a sensing slot of fixed duration $c_s$ within each update cycle of period $T$, during which communication is paused. Updating too frequently ($T \to c_s$) wastes resources on sensing without communication; updating too infrequently ($T \to \infty$) allows the twin to become stale, degrading FIM gain. We formalize this tradeoff as follows.

\begin{definition}[Effective FIM Gain]
\label{def:J}
The effective FIM gain over one DT cycle of period $T$ is
\begin{equation}
J(T) \triangleq \Bigl(1 - \frac{c_s}{T}\Bigr) \cdot \bar{G}(T),\quad \bar{G}(T) \triangleq \frac{1}{T} \int_0^T G(\tau)\, d\tau,
\label{eq:J_def}
\end{equation}
where $(1 - c_s/T)$ is the communication efficiency and $\bar{G}(T)$ is the time-averaged system FIM gain over the staleness interval.
\end{definition}

\begin{proposition}[Optimal DT Synchronization Period]
\label{prop:sync}
% Under the linear staleness approximation $G(\tau) \approx G_0 - \kappa\,\sigma_v^{(\mathrm{mode})}\,\tau$ valid for $\tau \in [0, \tau_{\mathrm{lin}})$ with $\tau_{\mathrm{lin}} \triangleq G_0/(\kappa\sigma_v^{(\mathrm{mode})})$, where $G_0 \triangleq G(0)$ is the nominal system FIM gain and $\kappa > 0$ is the per-unit-time degradation rate absorbing AP geometry, link count, and perpendicular projection factors from~\eqref{eq:angular_drift}, the effective FIM gain $J(T)$ admits a unique interior maximum on $(c_s, \tau_{\mathrm{lin}})$ at
Under the linear staleness approximation $G(\tau) \approx G_0 - 
\kappa\,\sigma_v^{(\mathrm{mode})}\,\tau$ valid for $\tau \in [0, 
\tau_{\mathrm{lin}})$ with $\tau_{\mathrm{lin}} \triangleq G_0/(\kappa
\sigma_v^{(\mathrm{mode})})$, where $G_0 \triangleq G(0)$ is the 
nominal system FIM gain and $\kappa > 0$ is the degradation rate 
absorbing AP geometry, link count, and projection factors 
from~\eqref{eq:angular_drift}, the effective FIM gain $J(T)$ admits 
a unique interior maximum on $(c_s, \tau_{\mathrm{lin}})$ at
\begin{equation}
T^\star = \sqrt{\frac{2\, c_s\, G_0}{\kappa\, \sigma_v^{(\mathrm{mode})}}}.
\label{eq:T_star}
\end{equation}
$T^\star$ scales as $(\sigma_v^{(\mathrm{mode})})^{-1/2}$ (poorer velocity tracking requires more frequent updates), $c_s^{1/2}$ (higher sensing cost tolerates longer intervals), and $G_0^{1/2}$; it is independent of the FIM morphing range~$\tilde{\zeta}$, which affects $G_0$ but not $\kappa$.
\end{proposition}

\begin{IEEEproof}
Under the linear approximation, $\bar{G}(T) = G_0 - \kappa\sigma_v^{(\mathrm{mode})}T/2$. Substituting into~\eqref{eq:J_def} gives
$J(T) = (1 - c_s/T)(G_0 - \kappa\sigma_v^{(\mathrm{mode})} T/2)$.
Differentiation yields $dJ/dT = c_s G_0/T^2 - \kappa\sigma_v^{(\mathrm{mode})}/2$, vanishing uniquely at $T^\star$ in~\eqref{eq:T_star}. The second derivative $d^2J/dT^2 = -2c_s G_0/T^3 < 0$ for all $T > 0$, so $J$ is strictly concave on $(c_s, \tau_{\mathrm{lin}})$ and $T^\star$ is the unique maximum. Scaling relations follow by inspection of~\eqref{eq:T_star}.
\end{IEEEproof}

\begin{remark}[Physical interpretation]
\label{rmk:Tstar_physics}
A vehicular target ($\|\mathbf{v}\| = 20$~m/s, $\sigma_v^{(\mathrm{angle})} \approx 1.1$~m/s) requires $T^\star \approx 5.6$~s under angle-only tracking while a pedestrian ($\|\mathbf{v}\| = 2$~m/s, $\sigma_v^{(\mathrm{angle})} \approx 0.56$~m/s) allows $T^\star \approx 15.5$~s. The relatively weak dependence on target speed reflects that the angle-only EKF \emph{does} estimate velocity, and it simply does so with larger residual error at higher speeds. The independence from $\tilde{\zeta}$ means investing in a wider morphing range increases the \emph{nominal} gain $G_0$ but does not change the \emph{optimal update rate}: DT synchronization design is decoupled from hardware specification.
\end{remark}
% \begin{remark}[Doppler-enhanced DT relaxes the update rate]
% \label{rmk:doppler_relax}
% Substituting the two measurement modes into~\eqref{eq:T_star}, the relaxation factor is
% $T^\star_{\mathrm{Dopp}}/T^\star_{\mathrm{angle}} = \sqrt{\sigma_v^{(\mathrm{angle})}/\sigma_v^{(\mathrm{Dopp})}} \approx \sqrt{0.8/0.15} \approx 2.3$ (Fig.~\ref{fig:dt}b). This reduction in update frequency translates directly to improved communication throughput, confirming that Doppler sensing is not merely an additional output but an enabler of efficient DT operation.
% \end{remark}

% \label{rmk:doppler_relax}
% Substituting the two measurement modes into~\eqref{eq:T_star}, the relaxation factor $T^\star_{\mathrm{Dopp}}/T^\star_{\mathrm{angle}} = \sqrt{\sigma_v^{(\mathrm{angle})}/\sigma_v^{(\mathrm{Dopp})}}$ grows with target speed: $1.3\times$ at $|\mathbf{v}|=2$~m/s, $2.4\times$ at $10$~m/s, and $3.6\times$ at $20$~m/s (Fig.~\ref{fig:dt_sync}b). The relaxation is most pronounced for fast-moving targets, where angle-only velocity inference is least accurate---a regime where DT efficiency matters most. This confirms that Doppler sensing is not merely an additional output but an enabler of efficient DT operation.
% \end{remark}

\begin{remark}[Doppler-enhanced DT relaxes the update rate]
\label{rmk:doppler_relax}
Substituting the two measurement modes into~\eqref{eq:T_star}, the relaxation factor $T^\star_{\mathrm{Dopp}}/T^\star_{\mathrm{angle}} = \sqrt{\sigma_v^{(\mathrm{angle})}/\sigma_v^{(\mathrm{Dopp})}}$ grows with target speed: $1.3\times$ at $|\mathbf{v}|=2$~m/s, $2.4\times$ at $10$~m/s, and $3.6\times$ at $20$~m/s (Fig.~\ref{fig:dt_sync}b). The relaxation is most pronounced for fast-moving targets, where angle-only velocity inference is least accurate, a regime where DT efficiency matters most. This confirms that Doppler sensing is not merely an additional output but an enabler of efficient DT operation.
\end{remark}

\subsection{Prediction Horizon Phase Transition}
\label{sec:horizon}

\begin{proposition}[Prediction Horizon]
\label{prop:horizon}
Define the prediction horizon $\tau_{\max}$ as the maximum pure-prediction time for which $\mathrm{RMSE}_p(\tau) \leq \epsilon$ for a given threshold~$\epsilon$. Then:
\begin{enumerate}[]
\item \emph{$M = 1$ (single-BS):} $\tau_{\max}$ is severely limited by the rank-1 velocity FIM, which leaves the tangential velocity unobservable.
\item \emph{$M \geq 2$ (cell-free):} the rank transition (Theorem~\ref{thm:rank}) enables full 2D velocity estimation, yielding $\tau_{\max} \propto \sqrt{\epsilon / \lambda_{\max}(\mathbf{P}_{vv})}$, which is finite and well-conditioned.
\item \emph{$M > 2$:} $\tau_{\max}$ saturates, with diminishing returns from additional APs.
\end{enumerate}
\end{proposition}

\begin{IEEEproof}
The position prediction error covariance after $\tau$ steps of pure prediction is $\mathbf{P}(\tau) = \mathbf{F}^\tau \mathbf{P}(0) (\mathbf{F}^\tau)^\top + \sum_{i=0}^{\tau-1}\mathbf{F}^i\mathbf{Q}(\mathbf{F}^i)^\top$. Extracting the position subblock and noting that $[\mathbf{F}^\tau]_{pp} = \mathbf{I}_2$, $[\mathbf{F}^\tau]_{pv} = \tau\Delta t\,\mathbf{I}_2$, the dominant term at large $\tau$ is $\mathbf{P}_{pp}(\tau) \approx \tau^2\Delta t^2\, \mathbf{P}_{vv}(0)$. For $M = 1$, $\mathbf{P}_{vv}(0)$ has one infinite eigenvalue (the tangential direction is unobservable from rank-1 velocity FIM), so $\mathrm{RMSE}_p(\tau)$ grows unboundedly regardless of $\tau$. For $M \geq 2$, $\mathbf{P}_{vv}(0)$ is finite and well-conditioned, giving $\tau_{\max} \propto \sqrt{\epsilon / \lambda_{\max}(\mathbf{P}_{vv})}$. Numerically, the prediction horizon increases by approximately $8\times$ at the $M=1\to 2$ transition, as a direct consequence of removing the infinite eigenvalue (cf.\ Fig.~\ref{fig:dt}).
\end{IEEEproof}

Incorporating Doppler measurements into the EKF provides two complementary benefits beyond the prediction horizon extension: (i) direct velocity observation, avoiding the need to infer velocity from position differences across time steps; and (ii) position refinement through the Doppler--position coupling term $\FIM_{\vect{pp}}^{(\mathrm{Doppler})}$ in~\eqref{eq:Ipp_Doppler} (Proposition~\ref{prop:joint_fim}). Section~\ref{sec:results} quantifies these benefits, showing a $2.3\times$ relaxation in DT update rate and decomposing the $23$~dB tracking gain into contributions from angular diversity, Doppler fusion, and FIM morphing.

%% file: 07_results.tex
\section{Numerical Results}
\label{sec:results}
%======================================================================
We validate the theoretical findings via Monte Carlo simulations with parameters in Table~\ref{tab:params}. All CRB values are averaged over $N_{\mathrm{mc}} = 200$ independent target realizations uniformly distributed inside the AP deployment region. Four baseline schemes are compared throughout: \emph{CF+FIM} (proposed), \emph{CF+Fixed}, \emph{SB+FIM}, and \emph{SB+Fixed}.
\begin{table}[h!]
\centering
\caption{Simulation parameters.}
\label{tab:params}
\begin{tabular}{l c c}
\hline
\textbf{Parameter} & \textbf{Symbol} & \textbf{Value} \\
\hline
Carrier frequency       & $f_c$                      & 28~GHz       \\
Wavelength              & $\lambda$                  & 10.71~mm     \\
Tx/Rx elements per AP   & $N_t,\,N_r$                & 8            \\
Tx/Rx AP count          & $M_t,\,M_r$                & 4            \\
FIM morphing range      & $\tilde{\zeta},\,\tilde{\xi}$ & $0.5\lambda$ \\
AP deployment radius    & $R_{\mathrm{AP}}$          & 40~m         \\
Waveform cross-corr.    & $\rho_w$                   & 0.3          \\
CPI duration            & $T_{\mathrm{cpi}}$         & 1~ms         \\
Target velocity         & $\mathbf{v}$               & $[3,\,4]^\top$~m/s \\
Process noise PSD       & $q$                        & 1.0~m$^2$/s$^3$ \\
% Monte Carlo runs        & $N_{\mathrm{mc}}$          & 200          \\
% SNR range / default     & ---                        & $[-10,\,30]$~dB / 10~dB \\
\hline
\end{tabular}
\end{table}

\subsection{CRB Validation}
\label{sec:res_crb}

\subsubsection{Rank Transition (Theorem~\ref{thm:rank})}
\vspace{-8pt}
\begin{figure}[ht]
\centering
\captionsetup{font=footnotesize}
\includegraphics{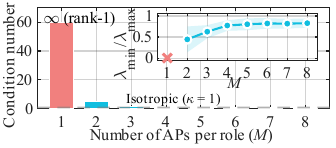}
\vspace{-6pt}
\caption{Velocity FIM rank transition vs.\ number of APs $M$%: $\kappa \to \infty$ at $M=1$ (rank-1, radial only), dropping to $\kappa \in [1.23, 1.33]$ for $M \geq 4$ (rank-2, full 2D). Inset shows the eigenvalue ratio converging to $0.83$.
}
\label{fig:rank}
\end{figure}
Fig.~\ref{fig:rank} validates Theorem~\ref{thm:rank}. For $M = 1$, the velocity FIM has $\kappa \to \infty$, confirming that only radial velocity is estimable when all bisectors collapse to a single direction. As soon as $M \geq 2$, the condition number drops to $\kappa = 5.1$ at $M=2$ and further to $\kappa \approx 1.23$--$1.33$ for $M \geq 4$, with the eigenvalue ratio stabilizing at $0.77$--$0.83$. The transition is abrupt: the minimum cell-free configuration ($M = 2$) already captures the qualitative rank-2 advantage; further APs refine accuracy without changing the rank.

\subsubsection{Velocity RMSE Scaling (Theorem~\ref{thm:scaling})}
%\vspace{-2pt}
\begin{figure}[ht]
\centering
\captionsetup{font=footnotesize}
\includegraphics{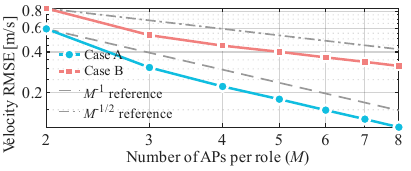}
\vspace{-6pt}
\caption{Velocity RMSE vs.\ $M$ %measured slopes $-1.03$ (Case A, matching $M^{-1}$) and $-0.53$ (Case B, matching $M^{-1/2}$).
}
\label{fig:vel_scaling}
\end{figure}
Fig.~\ref{fig:vel_scaling} validates the scaling laws in Theorem~\ref{thm:scaling}. Under Case~A (fixed per-pair SNR), the measured slope is $-1.03$, matching $M^{-1}$; under Case~B (fixed total power), the slope is $-0.53$, matching $M^{-1/2}$. The distinction arises because Case~A benefits from both additional measurements and maintained per-pair quality while Case~B trades per-pair SNR for spatial diversity. Position RMSE scales as $M^{-1.3}$, steeper than velocity's $M^{-1.0}$, due to a double benefit: more angular measurements and improved Jacobian conditioning.

\subsubsection{Spatial-Temporal Orthogonality (Theorem~\ref{thm:ortho})}
%\vspace{-2pt}
\begin{figure}[ht]
\centering

\captionsetup{font=footnotesize}\includegraphics{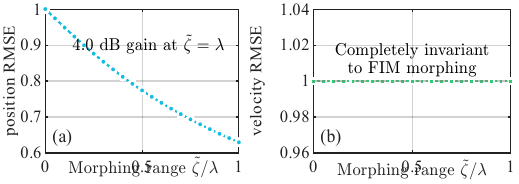}
\vspace{-18pt}
\caption{Normalized CRB vs.\ FIM morphing range $\tilde{\zeta}/\lambda$%: position CRB decreases by $4.0$~dB at $\tilde{\zeta}=\lambda$, while velocity CRB variation $<10^{-4}$ (Theorem~\ref{thm:ortho}).
}
\label{fig:orthogonality}
\end{figure}
Fig.~\ref{fig:orthogonality} validates Theorem~\ref{thm:ortho}. Position CRB decreases monotonically with $\tilde{\zeta}$, achieving 4.0~dB gain at $\tilde{\zeta} = \lambda$, while velocity CRB remains exactly constant (variation $< 10^{-4}$) across the entire morphing range. This invariance confirms that the velocity FIM depends on the bistatic bisectors $\{\vect{b}_{m_t,m_r}\}$, not on element displacements, through the temporal phase progression $e^{j2\pi f_d lT_s}$. The practical implication is that FIM shape can be optimized purely for position without any velocity penalty, eliminating the need for a position--velocity tradeoff in metasurface design.

\subsubsection{Mismatch Analysis (Theorem~\ref{thm:mismatch})}
\vspace{-8pt}
\begin{figure}[ht]
\centering 
\captionsetup{font=footnotesize}
\includegraphics[width=\linewidth]{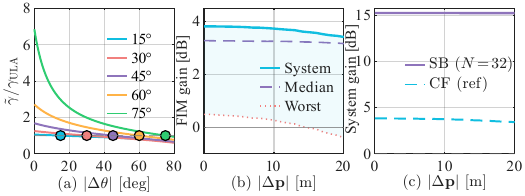}
\vspace{-15pt}
\caption{Mismatch robustness under DT prediction error. %(a)~per-link aperture ratio vs.\ angular mismatch; (b)~cell-free system gain vs.\ position error; (c)~single-BS contrast.
}
\label{fig:mismatch}
\end{figure}
% Fig.~\ref{fig:mismatch} validates Theorem~\ref{thm:mismatch} across three complementary views. Panel~(a) confirms that for each nominal angle $\hat{\theta}$, the normalized effective aperture $\tilde{\gamma}/\gamma_{\mathrm{ULA}}$ drops below unity at exactly $|\Delta\theta| = |\hat{\theta}|$, matching the closed-form critical mismatch~\eqref{eq:mm_delta_crit_exact}: at this point the true angle reaches broadside where element displacement produces zero path-length difference. Panel~(b) quantifies the system-level impact under cell-free deployment. The aggregate localization gain degrades gracefully from $3.9$~dB to $3.5$~dB as the DT position error grows to $20$~m. Although the worst per-link gain turns negative at $|\Delta\mathbf{p}| \approx 16$~m (exceeding $\Delta\theta_{\mathrm{crit}} = |\hat\theta|$ for the closest-to-broadside AP, Theorem~\ref{thm:mismatch}), the system-level gain remains strictly positive, confirming that angular diversity from $M=4$ distributed APs rescues overall performance. Panel~(c) provides the critical contrast: a single-BS deployment with equal total antenna count ($N=32$) lacks this diversity, and its system gain collapses far more rapidly, eventually falling below the ULA baseline. This comparison highlights that the robustness in panel~(b) arises specifically from the \emph{geometric diversity of observation angles} unique to cell-free architectures, not merely from averaging more links.
Fig.~\ref{fig:mismatch} validates Theorem~\ref{thm:mismatch} across three complementary views. 
%\emph{Panel~(a), per-link aperture ratio vs.\ angular mismatch}: for each nominal angle $\hat{\theta}$, the normalized effective aperture $\tilde{\gamma}/\gamma_{\mathrm{ULA}}$ drops below unity at exactly $|\Delta\theta| = |\hat{\theta}|$, matching the closed-form critical mismatch~\eqref{eq:mm_delta_crit_exact}: at this point the true angle reaches broadside where element displacement produces zero path-length difference. 
\emph{Panel (a), per-link aperture ratio versus angular mismatch}: for each nominal angle $\hat{\theta}$, the normalized effective aperture $\tilde{\gamma}/\gamma_{\mathrm{ULA}}$ falls below unity exactly at $|\Delta\theta| = |\hat{\theta}|$, consistent with the closed-form critical mismatch from Theorem~\ref{thm:mismatch}. At this point, the true angle reaches broadside, where element displacement induces zero path-length difference.
\emph{Panel~(b), cell-free system gain vs.\ position error}: the aggregate localization gain degrades gracefully from $3.9$~dB to $3.5$~dB as the DT position error grows to $20$~m. Although the worst per-link gain turns negative at $|\Delta\mathbf{p}| \approx 16$~m (exceeding $\Delta\theta_{\mathrm{crit}} = |\hat\theta|$ for the closest-to-broadside AP, Theorem~\ref{thm:mismatch}), the system-level gain remains strictly positive, confirming that angular diversity from $M=4$ distributed APs rescues overall performance. \emph{Panel~(c), single-BS contrast}: a single-BS deployment with equal total antenna count ($N=32$) lacks this diversity, and its system gain collapses far more rapidly, eventually falling below the ULA baseline. This comparison highlights that the robustness in panel~(b) arises specifically from the \emph{geometric diversity of observation angles} unique to cell-free architectures, not merely from averaging more links.

\subsection{ISAC Pareto and DT Tracking}
\label{sec:res_isac_dt}

\subsubsection{ISAC Pareto Tradeoff}
\begin{figure}[ht]
\centering
\vspace{-8pt}
\captionsetup{font=footnotesize}
\includegraphics{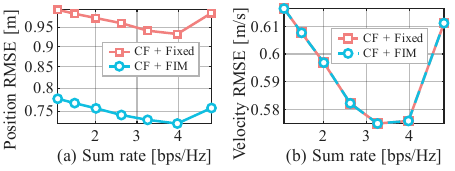}
\vspace{-4pt}
\caption{ISAC Pareto frontier%: position and velocity sensing RMSE vs.\ sum-rate.
}
\label{fig:pareto}
\end{figure}
% Fig.~\ref{fig:pareto} presents the ISAC Pareto frontier obtained by sweeping the per-user SINR threshold $\Gamma_k$, extending the MRT-plus-isotropic beamforming framework of~\cite{he2026globecom} (originally applied to position-only CRB) to the joint position--velocity setting. Each marker corresponds to a different $\Gamma_k$ value averaged over $80$ random target positions. In panel~(a), the CF+FIM frontier lies strictly below its CF+Fixed counterpart across the entire rate range, with a gap of $1.8$--$2.3$~dB that is widest under tight communication budgets (low sum-rate) and narrows as the SINR constraint relaxes. This is consistent with the aperture-gain analysis in Theorem~\ref{thm:mismatch}: when most power is devoted to sensing, the full morphing advantage is realized; when communication consumes the majority of the power budget, the residual sensing power diminishes and the absolute FIM improvement shrinks. Panel~(b) validates Theorem~\ref{thm:ortho} in the optimization context: the two velocity Pareto curves are visually indistinguishable, confirming that FIM shape optimization has \emph{zero} effect on velocity CRB regardless of the communication--sensing power split. This decoupling has a practical design implication: position and velocity estimation can be optimized through independent subsystems---FIM morphing for position and bistatic geometry for velocity---without cross-interference, simplifying the overall ISAC resource allocation.
Fig.~\ref{fig:pareto} presents the ISAC Pareto frontier obtained by sweeping the per-user SINR threshold $\Gamma_k$, extending the MRT-plus-isotropic beamforming framework of~\cite{he2026globecom} (originally applied to position-only CRB) to the joint position--velocity setting. Each marker averages over $80$ random target positions. In panel~(a), the CF+FIM frontier lies strictly below CF+Fixed with a gap of $1.8$ to $2.3$~dB, widest under tight communication budgets and narrowing as the SINR constraint relaxes; this is consistent with Theorem~\ref{thm:mismatch}, as the full morphing advantage is realized only when sufficient power is devoted to sensing. In panel~(b), the two velocity Pareto curves are visually indistinguishable, confirming that FIM shape optimization has \emph{zero} effect on velocity CRB regardless of the communication–sensing power split, validating Theorem~\ref{thm:ortho} in the optimization context. This decoupling enables independent subsystem design: FIM morphing for position, bistatic geometry for velocity, without cross-interference.

\subsubsection{DT Closed-Loop Tracking}
\vspace{-8pt}
\begin{figure}[ht]
\centering 
\captionsetup{font=footnotesize}
\includegraphics[width=\columnwidth]{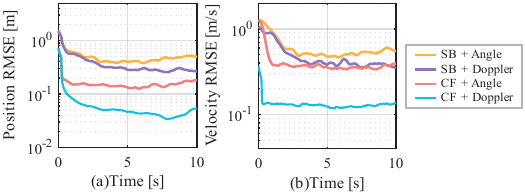}
\vspace{-16pt}
\caption{Closed-loop EKF tracking%: position (a) and velocity (b) RMSE vs.\ time.
}
\label{fig:dt}
\end{figure}
% Fig.~\ref{fig:dt} evaluates the EKF-based DT tracking loop of Section~\ref{sec:dt_loop} under four configurations that isolate the effects of cell-free deployment and Doppler fusion. All schemes converge from their noisy initialization to within $10\%$ of steady-state RMSE by $t \approx 2$~s. At steady state, CF+Doppler achieves the lowest position RMSE (${\sim}0.04$~m), followed by CF+Angle-only (${\sim}0.2$~m), SB+Doppler (${\sim}0.3$~m), and SB+Angle-only (${\sim}0.5$~m). The $4\times$ gap between CF+Doppler and CF+Angle-only shows that Doppler measurements tighten position estimates through the EKF cross-coupling between velocity and position states. The order-of-magnitude gap between CF and SB configurations at the same Doppler setting confirms that angular diversity is the dominant driver of position accuracy. Panel~(b) reveals a complementary hierarchy for velocity: the two Doppler-equipped schemes converge to ${\sim}0.13$~m/s regardless of whether the architecture is cell-free or single-BS, while the angle-only variants settle near $0.4$~m/s. This decoupling---velocity governed by Doppler availability, position by angular diversity---is consistent with the spatial-temporal orthogonality of Theorem~\ref{thm:ortho}: FIM shape and AP geometry affect position and velocity through independent subspaces of the joint FIM.
Fig.~\ref{fig:dt} evaluates the EKF tracking loop of Section~\ref{sec:dt_loop} under four configurations isolating the effects of cell-free deployment and Doppler fusion. All schemes converge within $2$~s to $10\%$ of steady-state RMSE. In panel~(a), steady-state position RMSE follows CF+Doppler $({\sim}0.04$~m) $<$ CF+Angle-only $({\sim}0.2$~m) $<$ SB+Doppler $({\sim}0.3$~m) $<$ SB+Angle-only $({\sim}0.5$~m). The $4\times$ gap between CF+Doppler and CF+Angle-only shows that Doppler tightens position through the EKF velocity-position cross-coupling; the order-of-magnitude gap between CF and SB at the same Doppler setting confirms that angular diversity is the dominant driver of position accuracy. Panel~(b) shows a complementary hierarchy for velocity: Doppler-equipped schemes converge to ${\sim}0.13$~m/s regardless of architecture, while angle-only variants settle near $0.4$~m/s. This decoupling, velocity governed by Doppler and position by angular diversity, is consistent with the spatial-temporal orthogonality of Theorem~\ref{thm:ortho}.

\subsection{Robust Design and DT Synchronization}
\label{sec:res_robust_sync}
This subsection addresses robust FIM design under DT uncertainty (Fig.~\ref{fig:robust}) and DT synchronization scheduling (Fig.~\ref{fig:dt_sync}), which follows the classical overhead-versus-performance tradeoff in adaptive transmission~\cite{Yu2025DTSync,zhang2011training}.
\begin{figure}[t]
\centering 
\captionsetup{font=footnotesize}
\includegraphics[width=\linewidth]{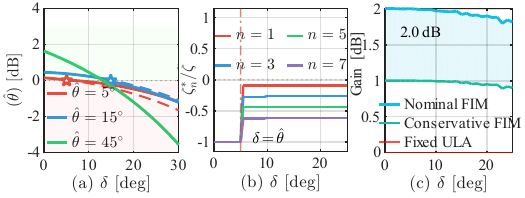}
\vspace{-18pt}
\caption{Robust FIM design under DT angle uncertainty $\delta$.}
\label{fig:robust}
\end{figure}
%Fig.~\ref{fig:robust} illustrates the robust minimax FIM design from three perspectives. \emph{Panel~(a), per-link gain}: the FIM gain $\tilde{\gamma}(\hat{\theta})$ relative to a fixed ULA crosses the baseline at $\delta \approx |\hat\theta|$ for all three tested angles $(5^\circ, 15^\circ, 45^\circ)$, matching the closed-form prediction $\Delta\theta_{\mathrm{crit}} = |\hat{\theta}|$ of Theorem~\ref{thm:mismatch}; beyond the crossing the gain turns negative, with the degradation weaker at small $\hat\theta$ because near broadside, element displacement produces negligible path-length difference.
Fig.~\ref{fig:robust} illustrates the robust minimax FIM design from three perspectives. \emph{Panel (a), per-link gain}: the FIM gain $\tilde{\gamma}(\hat{\theta})$ relative to a fixed ULA crosses the baseline at $\delta \approx |\hat{\theta}|$ for all tested angles $(5^\circ, 15^\circ, 45^\circ)$, consistent with the closed-form prediction $\Delta\theta_{\mathrm{crit}} = |\hat{\theta}|$ in Theorem~\ref{thm:mismatch}. Beyond this point, the gain becomes negative; the degradation is less pronounced for small $\hat{\theta}$, as near-broadside operation results in negligible path-length variation due to element displacement. The robust design (dashed) degrades more gracefully but at the cost of reduced peak gain, tracking the ULA baseline rather than dropping below it.
%The robust design (dashed) degrades more gracefully but sacrifices peak gain, tracking the ULA baseline rather than falling below it. 
\emph{Panel (b), optimal displacement}: for $\delta < \hat{\theta}$ (with $\hat{\theta} = 5^\circ$), all elements lie on the morphing boundary, i.e., $\zeta_n^{\mathrm{opt}} = \pm \tilde{\zeta}$, in order to maximize the effective aperture. At $\delta = \hat{\theta}$, the elements shift to interior positions, confirming the bang-bang structure of the minimax solution. \emph{Panel (c), system-level gain} ($M_t = M_r = 4$): the nominal design maintains an approximately $2.0$~dB gain over the ULA baseline across the uncertainty range (starting from $2.01$~dB at $\delta = 0$), while the conservative design remains at approximately $1.0$~dB (starting from $1.00$~dB at $\delta = 0$). Neither curve crosses the $0$~dB baseline, confirming that angular diversity from distributed APs prevents the per-link inversions observed in Panel (a) from propagating to the system level.
\subsubsection{DT Synchronization Rate 
(Proposition~\ref{prop:sync})}
\begin{figure}[ht]
\centering 
\captionsetup{font=footnotesize}
\includegraphics[width=\linewidth]{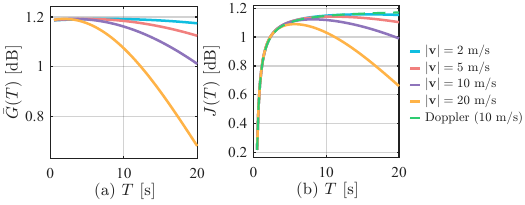}
\vspace{-16pt}
\caption{DT synchronization%: time-averaged FIM gain $\bar{G}(T)$ (a) and effective gain $J(T)$ (b) vs.\ update period $T$.
}
\label{fig:dt_sync}
\end{figure}
Fig.~\ref{fig:dt_sync} validates Proposition~\ref{prop:sync}. \emph{Panel~(a), time-averaged gain}: all four speed curves start from $\bar{G}(0^+) \approx 1.2$~dB but diverge as $T$ grows; at $|\mathbf{v}|=20$~m/s, angular drift $\omega_m = \sigma_v^{(\mathrm{mode})}/r_m$ accumulates rapidly and $\bar{G}$ drops below $1$~dB by $T \approx 12.5$~s, whereas at $2$~m/s the gain stays above $1.1$~dB even at $T=20$~s. \emph{Panel~(b), effective gain}: the sensing-slot overhead $c_s$ penalizes short update periods via $(1-c_s/T)$, yielding a concave $J(T)$ with a well-defined peak for each speed. Under angle-only tracking, $T^\star_{\mathrm{angle}}$ shifts from $5.6$~s ($|\mathbf{v}|=20$~m/s) through $8.2$~s ($10$~m/s) and $10.9$~s ($5$~m/s) to $15.5$~s ($2$~m/s), confirming the $T^\star \propto \sigma_v^{-1/2}$ scaling of~\eqref{eq:T_star}. Doppler fusion extends $T^\star_{\mathrm{Dopp}}$ to ${\approx}20$~s across all speeds, reducing the velocity residual from $\sigma_v^{(\mathrm{angle})} \approx 0.6$ to $1.2$~m/s down to $\sigma_v^{(\mathrm{Dopp})} \approx 0.15$~m/s. The resulting relaxation factor grows from $1.3\times$ at $|\mathbf{v}|=2$~m/s to $3.6\times$ at $|\mathbf{v}|=20$~m/s, a direct benefit of Doppler fusion on the DT update budget.

\subsection{Comparison with Existing DT-ISAC Methods}
\label{sec:res_benchmark}

To quantify the system-level benefits, we compare against a baseline inspired by~\cite{Ding2024DTISAC}, which employs a co-located BS with fixed ULA, angle-only EKF tracking, and fixed-period DT updates. For fairness, the baseline uses the same total antenna count ($N = 32$) and the same constant-velocity EKF. We progressively enable each proposed module in five configurations: (i)~\emph{Baseline~\cite{Ding2024DTISAC}}: single-BS, fixed ULA, angle-only, fixed update period; (ii)~\emph{+Cell-free}: distributed APs ($M_t=M_r=4$), fixed ULA, angle-only; (iii)~\emph{+FIM}: cell-free with FIM morphing, angle-only; (iv)~\emph{+Doppler}: cell-free, FIM, angle + bistatic Doppler; (v)~\emph{Proposed}: full framework with adaptive DT scheduling (Algorithm~\ref{alg:policy}).

\begin{figure}[ht]
\centering 
\captionsetup{font=footnotesize}
\includegraphics[width=\linewidth]{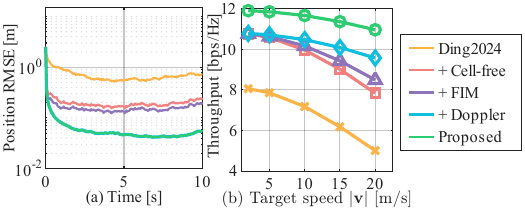}
\vspace{-16pt}
\caption{Progressive ablation benchmark%: position RMSE vs.\ time (a) and effective throughput vs.\ target speed (b).
}
\label{fig:benchmark}
\end{figure}

\subsubsection{Tracking Accuracy}
Fig.~\ref{fig:benchmark}(a) shows position RMSE over a $10$~s tracking window. The baseline~\cite{Ding2024DTISAC} converges to $0.69$~m, limited by its single observation geometry and lack of Doppler. Cell-free deployment reduces this to $0.21$~m ($-10.2$~dB), confirming that angular diversity is the single largest contributor to position accuracy. Adding FIM morphing gives a further $1.9$~dB reduction to ${\sim}0.17$~m, a modest improvement consistent with Theorem~\ref{thm:ortho}, since FIM affects only the position sub-CRB while velocity (which drives EKF steady-state position via cross-coupling) is unaffected. Doppler fusion reduces the error to ${\sim}0.05$~m, contributing 10.9 dB 
to the total tracking gain, validating the asymmetric coupling predicted by Proposition~\ref{prop:asym}: Doppler improves velocity estimation, which in turn tightens position via the EKF cross-covariance. The total gap from baseline is $23$~dB. Configurations (iv) and (v) produce identical tracking curves because Algorithm~\ref{alg:policy} does not reduce the sensing rate; instead, it optimizes \emph{when} to refresh the DT, reducing communication overhead without sacrificing tracking fidelity.%Configurations~(iv) and~(v) produce identical tracking curves because Algorithm~\ref{alg:policy} does not reduce the sensing rate; it optimizes \emph{when} to refresh the DT, saving communication overhead without sacrificing tracking fidelity.

\subsubsection{Effective Throughput}
Effective throughput is defined as $\eta_{\mathrm{eff}} = (1 - c_s/T)\,R_{\mathrm{sum}}\,\alpha(|\mathbf{v}|)$, with $c_s = 0.1$~s the per-slot sensing overhead, $T$ the update period ($T=1$~s fixed, $T=T^\star$ adaptive), and $\alpha$ the beam-alignment loss from DT staleness. At $|\mathbf{v}|=20$~m/s, Fig.~\ref{fig:benchmark}(b) shows that the baseline delivers $5.0$~bps/Hz, limited by angle-only EKF and fixed scheduling. Cell-free deployment provides the largest single throughput gain ($+2.8$~bps/Hz, to $7.8$~bps/Hz) via macro-diversity; FIM morphing adds $+0.7$~bps/Hz; Doppler fusion, $+1.1$~bps/Hz, by maintaining beam alignment at high speed; adaptive scheduling adds $+1.4$~bps/Hz by reducing $c_s/T^\star$ below that of the fixed schedule, exploiting the EKF's low steady-state covariance to allow $T^\star \gg 1$~s. The proposed framework reaches $11.0$~bps/Hz, a $119\%$ gain over baseline.

Accuracy and throughput follow \emph{distinct rankings}. Tracking is dominated by Doppler fusion ($-10.9$~dB), then cell-free diversity ($-10.2$~dB) and FIM morphing ($-1.9$~dB); throughput is dominated by cell-free deployment ($+2.8$~bps/Hz), then adaptive scheduling ($+1.4$~bps/Hz), Doppler fusion ($+1.1$~bps/Hz), and FIM morphing ($+0.7$~bps/Hz). This asymmetry matches the theoretical predictions: Theorem~\ref{thm:rank} and Proposition~\ref{prop:asym} identify Doppler as the tracking enabler, while Proposition~\ref{prop:sync} predicts that adaptive scheduling's payoff is primarily in communication overhead.

%% file: 08_conclusion.tex
\section{Conclusion}
\label{sec:conclusion}

This paper developed a joint position--velocity CRB framework for FIM-augmented cell-free ISAC and a closed-loop tracking architecture that translates these bounds into practical design rules. Four main results were obtained.
First, a scatter decomposition of the velocity FIM reveals a rank transition at two APs: a single BS estimates only radial velocity, while two APs enable full 2D observability with RMSE scaling as $M^{-1}$.
Second, a spatial-temporal orthogonality principle shows that FIM shaping affects position CRB but not velocity CRB under isotropic waveforms, allowing decoupled design.
Third, mismatch analysis shows that performance degrades beyond a critical mismatch equal to the nominal pointing angle while cell-free diversity preserves system-level gains.
Fourth, the proposed closed-loop tracking with confidence-aware scheduling and optimal synchronization enables adaptive FIM control at each AP. Simulation demonstrates a 23~dB tracking gain over a co-located baseline, with additional improvements in synchronization and throughput.

Future work includes multi-target extension, near-field sub-THz propagation, tighter non-local bounds such as the Barankin bound~\cite{huang2026barankin} for threshold SNR prediction, and experimental validation on a cell-free testbed.

%% file: Appendix.tex
\appendices

\section{Derivation of the Robust FIM SDP}
\label{app:sdp}

By Theorem~\ref{thm:ortho}, the robust FIM shape optimization decouples across APs in the sense that the velocity CRB is shape-independent; the objective~\eqref{eq:robust_obj} thus reduces to minimizing the worst-case position CRB. Substituting the Taylor expansion~\eqref{eq:taylor_crb} into the epigraph reformulation of~\eqref{eq:robust_P0}, the per-AP worst-case constraint reads
\begin{equation}
h_{0,m} + h_{1,m}\,\Delta\theta_m + h_{2,m}\,\Delta\theta_m^2 \le t_m,\quad \forall\,\Delta\theta_m^2 \le \delta_m^2,
\label{eq:app_semi_inf}
\end{equation}
a semi-infinite quadratic constraint with a scalar quadratic uncertainty set. Applying the S-procedure~\cite[App.~B.2]{boyd2004convex} (Slater feasible since $\Delta\theta_m = 0$ is an interior point), constraint~\eqref{eq:app_semi_inf} is equivalent to the existence of $\tau_m \ge 0$ such that
\begin{equation}
(h_{2,m} - \tau_m)\Delta\theta_m^2 + h_{1,m}\Delta\theta_m + (h_{0,m} + \tau_m\delta_m^2 - t_m) \le 0,\;\forall\,\Delta\theta_m.
\label{eq:app_scalar_psd}
\end{equation}
A scalar quadratic $\alpha x^2 + \beta x + \gamma \le 0$ holds for all $x\in\mathbb{R}$ iff $\alpha \le 0$ and $\beta^2 \le 4\alpha\gamma$, which in matrix form is equivalent to $\bigl[\begin{smallmatrix}-\gamma & -\beta/2\\ -\beta/2 & -\alpha\end{smallmatrix}\bigr] \succeq \mathbf{0}$. Substituting the coefficients from~\eqref{eq:app_scalar_psd} yields the LMI~\eqref{eq:robust_lmi}, and combined with the epigraph objective gives the SDP~\eqref{eq:robust_sdp}. Since the S-procedure is \emph{exact} for scalar quadratic constraints, the SDP is an equivalent reformulation, not a relaxation, of~\eqref{eq:robust_P0} under the Taylor approximation~\eqref{eq:taylor_crb}. The resulting problem has $M_t + M_r$ size-$2$ LMIs and is solvable by standard interior-point methods with per-iteration complexity $\mathcal{O}((MN)^{3.5})$.\hfill$\blacksquare$

\section{Velocity FIM Derivation via Slepian--Bangs}
\label{app:vel_fim}

% Applying the Slepian--Bangs formula for deterministic signals in complex Gaussian noise~\cite{kay1993estimation} to the Doppler signal model~\eqref{eq:doppler_signal}, the velocity sub-block of the per-pair Fisher information is
Applying the Slepian--Bangs formula~\cite{kay1993estimation} for complex 
Gaussian observations to the Doppler model~\eqref{eq:doppler_signal}, 
the velocity sub-block of the per-pair Fisher information is
\begin{equation}
[\mathbf{I}_{\mathbf{vv}}^{(m_t,m_r)}]_{ij}
= \frac{2}{\sigma^2}\,\Real\bigg\{\sum_{l=0}^{L-1}
\frac{\partial\bm{\mu}_l^H}{\partial v_i}\frac{\partial\bm{\mu}_l}{\partial v_j}\bigg\}.
\label{eq:app_sb}
\end{equation}
Since, $\mathbf{v}$, enters $\bm{\mu}_l = g_{m_t,m_r}\, e^{j2\pi f_d l T_s}\,\mathbf{a}_{m_r}$ only through $f_d = \lambda^{-1}\mathbf{b}_{m_t,m_r}^\top\mathbf{v}$, where $g_{m_t,m_r} \triangleq \alpha_{m_t,m_r}\mathbf{a}_{m_t}^H\mathbf{w}_{m_t}$, the chain rule yields
\begin{equation}
\frac{\partial\bm{\mu}_l}{\partial v_i}
= \frac{j2\pi l T_s}{\lambda}\,[\mathbf{b}_{m_t,m_r}]_i\,g_{m_t,m_r}\,e^{j2\pi f_d l T_s}\,\mathbf{a}_{m_r}.
\label{eq:app_deriv}
\end{equation}
Substituting~\eqref{eq:app_deriv} into~\eqref{eq:app_sb}, the Doppler phase exponentials cancel and $\|\mathbf{a}_{m_r}\|^2 = 1$ by array normalization, giving a result independent of $f_d$ and of the Rx FIM displacement, $\bm{\xi}_{m_r}$. The Tx-side displacement $\bm{\zeta}_{m_t}$ is absorbed into $|g_{m_t,m_r}|^2$ via the beamforming gain $|\mathbf{a}_{m_t}^H\mathbf{w}_{m_t}|^2$, which is shape-independent under matched-filter or isotropic beamforming. These cancellations are the structural basis of Theorem~\ref{thm:ortho}. Evaluating $\sum_{l=0}^{L-1}l^2 = L(L-1)(2L-1)/6 \approx L^3/3$ for $L \gg 1$ and summing the resulting rank-1 outer products $\lambda^{-2}\eta^2_{m_t,m_r}\mathbf{b}_{m_t,m_r}\mathbf{b}_{m_t,m_r}^\top$ over all $M_tM_r$ bistatic pairs yields~\eqref{eq:vel_fim_total} with $\eta^2_{m_t,m_r}$ as defined therein.\hfill$\blacksquare$